\documentclass[journal]{IEEEtran}
\usepackage{amsmath}
\usepackage{bm}
\usepackage{amssymb}
\usepackage{graphicx}
\usepackage{color, colortbl} 
\usepackage[dvipsnames,svgnames,x11names]{xcolor}
\usepackage{algpseudocode,algorithm}
\usepackage{enumerate}
\usepackage{cite}
\usepackage{url,comment}
\usepackage[normalem]{ulem}
\usepackage[shortlabels]{enumitem}
\usepackage{amsthm}
\usepackage{pgfplots}
\usepackage{tikz}
\usepackage[font={small}]{caption}
\usetikzlibrary{arrows.meta}
\usetikzlibrary{calc}
\makeatletter
\newcommand{\gettikzxy}[3]{%
  \tikz@scan@one@point\pgfutil@firstofone#1\relax
  \edef#2{\the\pgf@x}%
  \edef#3{\the\pgf@y}%
}
\usepackage{mathtools}
\usepackage{balance}
\usepackage{bbm}
\usepackage{acronym}
\usepackage{caption}
\usepackage{subcaption}
\newcommand\numberthis{\addtocounter{equation}{1}\tag{\theequation}}
\newcommand{\appropto}{\mathrel{\vcenter{
  \offinterlineskip\halign{\hfil$##$\cr
    \propto\cr\noalign{\kern2pt}\sim\cr\noalign{\kern-2pt}}}}}

\newtheorem{proposition}{Proposition}

\acrodef{BS}{base station}
\acrodef{CSI}{channel state information}
\acrodef{DFT}{discrete Fourier transform}
\acrodef{UE}{user equipment}
\acrodef{ISAC}{integrated sensing and communication}
\acrodef{ISLAC}{integrated sensing, localization, and communication}
\acrodef{LOS}{line-of-sight}
\acrodef{RIS}{reconfigurable intelligent surface}
\bstctlcite{IEEEexample:BSTcontrol}

\pgfplotsset{compat=1.16}

\title{Two-Timescale Transmission Design and RIS Optimization for Integrated Localization and Communications}
\author{Fan Jiang, \IEEEmembership{Member, IEEE},  Andrea Abrardo, \IEEEmembership{Senior Member, IEEE}, Kamran Keykhoshravi, \IEEEmembership{Member, IEEE}, Henk Wymeersch, \IEEEmembership{Senior Member, IEEE}, Davide Dardari, \IEEEmembership{Senior Member, IEEE}, Marco Di Renzo, \IEEEmembership{Fellow, IEEE}

\thanks{This work has been supported, in part, by the European Commission through the H2020 RISE-6G project under grant agreement number 101017011 and the H2020 ARIADNE project under grant agreement number 871464. D. Dardari was sponsored in part by the Theory Lab, Central Research Institute, 2012 Labs, Huawei Technologies Co., Ltd.}

\thanks{F. Jiang is with School of Information Technology, Halmstad University, Sweden. Email: fan.jiang@hh.se}

\thanks{A. Abrardo is with Department of Information Engineering and Mathematical Sciences, University of 
Siena, Italy and with CNIT, National Inter-University Consortium for Telecommunications
Italy. Email: abrardo@unisi.it} 

\thanks{ K. Keykhoshravi is with Ericsson AB, 41756, Gothenburg, Sweden. Email: kamran.keykhosravi@ericsson.com}

\thanks{
H. Wymeersch is with Department of Electrical Engineering, Chalmers University of Technology, Sweden. Email: henkw@chalmers.se}

\thanks{D. Dardari is with Department of Electrical, Electronic, and Information Engineering, University of Bologna, Italy and with CNIT, National Inter-University Consortium for Telecommunications
Italy. Email: davide.dardari@unibo.it} 

\thanks{M. Di Renzo is with Université Paris-Saclay, CNRS, CentraleSupélec, Laboratoire des Signaux et Systémes. Email: marco.di-renzo@universite-paris-saclay.fr}
}

\begin{document}
\usetikzlibrary{shapes.multipart,intersections}

\maketitle
\begin{abstract}
    \Acp{RIS} have  tremendous potential to boost communication performance, especially when the \ac{LOS} path between the \ac{UE} and \ac{BS} is blocked. To control the \ac{RIS},  \ac{CSI} is needed, which entails significant pilot overhead. To reduce this overhead and the need for frequent \ac{RIS} reconfiguration, we propose a novel framework for integrated localization and communications, where \ac{RIS} configurations are fixed during location coherence intervals, while \ac{BS} precoders are optimized every channel coherence interval. This framework leverages accurate location information obtained with the aid of several \acp{RIS} as well as novel \ac{RIS} optimization and channel estimation methods. Performance in terms of localization accuracy, channel estimation error, and achievable rate demonstrates the effectiveness of the proposed approach. 
\end{abstract}

%---------------------------------------------
%---------------------------------------------
%---------------------------------------------

\section{Introduction}
%---------------------------------------------
%\subsubsection*{Background}
As the fifth-generation (5G) cellular network is being deployed worldwide, the research community is investigating key technologies towards the sixth-generation (6G), which is expected to be standardized in the late 2020s \cite{jiang2021road, UusRugBol21,YouWanHuaGao21}. Among the key enablers, we count the introduction of reconfigurable intelligent surfaces (\Acp{RIS}), which are large planar arrays of configurable small meta-atoms \cite{LiuLiuMuHou21,basar2019wireless,  RenZapDebAloYueRosTre20}. Such \Acp{RIS} can be placed on regular surfaces and through their configuration enable the modification of the radio propagation channel far beyond what was previously possible. An important canonical use case is to overcome the line-of-sight (\ac{LOS}) blockage between a base station (\ac{BS}) and a user equipment (\ac{UE}), which is especially relevant in millimeter wave (mmWave) and sub-terahertz (THz) frequency bands (30 GHz - 300 GHz)  \cite{RenZapDebAloYueRosTre20, wan2021terahertz,BouAle21}. 

In parallel to this technology-driven development, another key enabler towards 6G is the convergence of communication, localization and sensing, often referred to as \ac{ISAC} \cite{liu2021integrated} or even \ac{ISLAC} \cite{wymeersch2021integration,LimBelBerBou21}. Driven largely by the increased available bandwidth and larger arrays at the transmitter and receiver, radio signals have great potential to enable accurate localization and sensing \cite{BarWymMacBru21}. Moreover, the geometric nature of the propagation channel even mandates that communication, localization, and sensing should be jointly designed, as high-rate directional transmissions can only be provided with prior knowledge of \ac{UE} locations and predictions of blockages in the environment. This will make 6G the first generation where localization is not an add-on feature to communication, but localization is designed from the onset to operate jointly with communication, with strong mutual synergies \cite{wymeersch2021integration, HeJiaKeyKokWymJun21}. \acp{RIS} are expected to play an important role, both to increase or maintain data rates, but also to support accurate user localization and tracking \cite{BjoWymMatPop22,wymeersch2021integration}. 
Research on \acp{RIS} for communication and localization has progressed enormously over the past few years \cite{StrAleSciDi21,Dardari20, BjoWymMatPop22,TanCheCheDai21,LiuWuDiYua22,KeyKesGraWym20,ZhaZhaDiBiaHanSon20}. Nevertheless, important and fundamental challenges remain. Among these challenges, two inter-related problems stand out: (i) \emph{how is it possible to reduce the channel state information (\ac{CSI}) estimation overhead and the RIS configuration rate by jointly exploiting localization and communication?} (ii) \emph{How and when to control the \ac{RIS} meta-atoms to support communication and localization functions?} 
%---------------------------------------------
\subsubsection*{RIS Channel Estimation}
Conventionally, the RIS configuration requires the knowledge of \ac{CSI} over the BS-RIS and RIS-UE links. The channel estimation problem is challenging, as a nearly-passive \ac{RIS} does not have either the possibility to locally estimate the channel or the possibility of sending pilot signals \cite{ZheYouMeiZha22, LiaCheLonHeLinHuaLiuSheRen21}. This leads to a cascaded channel estimation problem at either the \ac{BS} or \ac{UE} \cite{LiaCheLonHeLinHuaLiuSheRen21}, which requires a sequence of different \ac{RIS} configurations during the estimation process. The number of such configurations depends on the underlying structure of the channel. When the channel is sparse and thus can be described by few parameters, efficient channel estimation is possible. In unstructured channels, on the other hand, the overhead is proportional to the number of \ac{RIS} elements if the RIS is optimized based on instantaneous \ac{CSI}.
In \cite{YouZheZha:20,JenCar:20}, for example, random phase and structured configurations of \ac{RIS} elements are investigated by showing how \ac{DFT}-based sequences achieve the minimum variance estimation provided that the number of pilot symbols is larger or equal to the number of RIS elements. Inspired by the sparsity of the mmWave/THz propagation channel, a class of compressed sensing (CS) based channel estimation schemes have been developed \cite{wan2021terahertz, WanFanDuaLi20, HeWymJun21, MirAli21}. Generally, the performance of CS-based methods relies on the proper sparse representation of the cascaded channel, and better estimation performance usually requires high computational complexity. Another strategy is to factorize the high-dimensional cascaded channel into a set of low-dimensional sub-channels \cite{WeiHuaAleYue21,LiuYuaZha20,LinJinMatYou21}. By recovering the low-dimensional subchannels through subspace-based algorithms (such as singular value decomposition (SVD) or high-order SVD), the overall cascaded channel is obtained. However, the training overhead and complexity increase with the size of all dimensions, and the channel estimation performance also depends on the accuracy of the decomposition models. The application of machine learning methods for RIS channel estimation have been considered \cite{KunMck21, LiuNgYua22,LiuLeiZha20}. Since analytical closed-form expressions of the channel estimate and channel estimation error are hard to obtain with machine learning, this approach is usually not suitable to theoretically quantify the RIS performance.

%---------------------------------------------
\subsubsection*{RIS Optimization in Communications}
In general, channel estimation for \ac{RIS} configuration requires significant amount of training overhead \cite{ZheYouMeiZha22,ZapRenShaQiaDeb21}. Instead of relying on instantaneous \ac{CSI} for RIS optimization, the use of statistical CSI for RIS configuration has been considered recently \cite{LuoLiJinChe21, GanZhoHuaZha21,JiaYeCui20,DaiZhuPanRenWan22}. For example, upper bounds for the uplink and downlink ergodic rate are derived in \cite{GanZhoHuaZha21,LuoLiJinChe21}. Based on the obtained analytical formulas, the RIS phase configuration is optimized using the alternating direction method of multipliers (ADMM) and alternating optimization (AO) based algorithms. The authors of \cite{JiaYeCui20} study a multiple BS interference channel, and derive an upper bound for the ergodic rate by assuming the maximal ratio transmission (MRT) precoding scheme. The RIS optimization is performed by using an iterative parallel coordinate descent (PCD) method. The research works in \cite{LuoLiJinChe21, GanZhoHuaZha21,JiaYeCui20} consider either single-user or signal-antenna user setups. In addition, the ergodic rate is obtained by assuming a specific transmission scheme (such as the MRT in \cite{GanZhoHuaZha21,DaiZhuPanRenWan22}). In order to improve the achievable rate, the two-timescale \ac{CSI} scheme \cite{HanTanJinWenMa19, HuDaiHanWan21,  ZhiPanRenWan21,Pan21} has been proposed. The main idea of the two-timescale scheme is to optimize the RIS configuration based on long-term \ac{CSI}, and the precoding at the BS based on instantaneous \ac{CSI} \cite{HanTanJinWenMa19,Pan21}. Since the long-term \ac{CSI} changes slowly, the RIS configuration does not need to be updated frequently. In addition, the acquisition of the instantaneous \ac{CSI} for optimizing the BS precoding requires a number of pilot symbols that depends on the number of users and antenna elements at that UEs, but that is independent of the number of RIS elements. In \cite{ZhiPanRenWan21}, the authors first obtain the aggregated channel from the users to the BS in each channel coherence interval, by utilizing a linear minimum mean square error estimator (LMMSE), and they then derive closed-form expressions of the ergodic achievable rate under the assumption of maximal ratio combining (MRC) at the BS. Due to these assumptions (LMMSE and MRC), the resulting expression of the uplink ergodic rate is not optimal \cite{ZhiPanRenWan21,Pan21}. The analysis in \cite{ZhiPanRenWan21} shows, however, that the two-timescale scheme requires knowledge of the locations and the angles of the users with respect to the BS and the RIS, which vary much slower than the instantaneous CSI. Specifically, the rate depends on the \emph{location}-dependent angle-of-arrival (AOA) and angle-of-departure (AOD), which indicates that accurate estimates of the locations of the users are needed for system optimization. By contrast, the location uncertainty of the UEs has been taken into account in \cite{abrardo2020intelligent} when optimizing the \ac{RIS} configuration. However, the authors of \cite{abrardo2020intelligent} did not consider the inherent ability of \acp{RIS} to help estimating the locations of the UEs. Also, the precoding design does not account for channel estimation errors.

%---------------------------------------------
\subsubsection*{RIS for Localization}
Based on these research works, it is apparent that the optimization of RIS-aided channels based on statistical CSI requires the estimation of the location of the UEs. The presence of RISs can help this challenging task. Specifically, localization theoretical bounds have been investigated in \cite{WymDen20,HeWymKonSilJun20, ElzGueGuiAlo20}  with the purpose of understanding the potential advantages of using RISs compared to schemes based only on the natural scattering of the environment. Practical RIS-aided localization algorithms can be found in \cite{ZhaZhaDiBiaHanSon20,ZhaZhaDiBiaHanSon21,NguGeoGra20,KeyKesGraWym20,BjoWymMatPop22,DarDecGueGui:J21}.
In \cite{ZhaZhaDiBiaHanSon20,ZhaZhaDiBiaHanSon21}, a receive signal strength (RSS) based multi-user positioning scheme is proposed in which the phase profile of the RIS is optimized to obtain a favorable RSS distribution in space, and thus a better discrimination of the RSS signature of neighboring locations.  
In \cite{NguGeoGra20}, a machine learning method for RSS-based  fingerprint localization is investigated. The authors demonstrate that the diversity offered by a RIS can be successfully used to generate reliable radio maps. 
Better performance can be obtained by exploiting phase/time-of-arrival (TOA) of signals. In this direction, the authors of  \cite{KeyKesGraWym20,BjoWymMatPop22} propose a low-complexity localization algorithm that estimates the TOA of the direct path and the path reflected by the RIS, as well as the AOD from the RIS to infer, in the far-field region, the position of the UE in the presence of synchronization errors. In \cite{DarDecGueGui:J21}, a narrowband and a two-step wideband positioning algorithms exploiting  near-far propagation conditions are proposed. These solutions operate when the BS is blocked and positioning relies only on the signal reflected by the RIS. To our best knowledge, there exist no research works that combine localization and communication with the purpose of reducing the overall signaling and estimation overhead in \ac{RIS}-aided networks.

In this paper, we investigate how communication and localization can be jointly exploited in order to significantly reduce the \ac{CSI} estimation overhead by optimizing the configuration of the \ac{RIS} for several channel coherence intervals. This leads to a novel frame structure, comprising infrequent localization and \ac{RIS} control tasks, combined with a more frequent optimization of the \ac{BS} precoders to maximize the transmission rate. Our contributions are summarized as follows:
\begin{itemize}
    \item We propose a novel integrated localization and communication framework and protocol for multi-RIS,  multi-user MIMO communications, consisting of three phases: Phase I for localization, Phase II for location-aided channel estimation, and Phase III for data transmission. Instead of relying on external sources for localization, the proposed framework obtains the location information based on the transmission of dedicated pilot signals in Phase I. Also, as the locations of the UEs change slowly compared to the channel variations, we perform the localization task of the UEs based on a longer timescale.
    \item We propose a design approach for the optimal RIS profile that requires the RIS configuration at a low rate and that accounts for the localization performance in Phase I. Resorting to the location-based RIS optimization scheme proposed in \cite{abrardo2020intelligent}, the proposed RIS design strategy works well both in static and dynamic scenarios.
    \item We design a channel estimation scheme in Phase II using location information. We derive a closed-form expression for the channel covariance matrix of the estimation error and show that it is sufficiently accurate with the aid of numerical results. Our analysis reveals that the effective achievable rate is significantly improved by the proposed channel estimation algorithm that relies on prior location information compared to existing approaches.
    \item We propose optimal precoder design schemes that account for the channel estimation error to maximize the {conditional achievable rate. The proposed precoder design achieves near-optimal rate performance with a small number of pilot symbols, thanks to} a proper configuration of the RIS phase profile, an improved channel estimation, and a proper precoder design, all benefiting from the localization in Phase I.
\end{itemize}
The proposed framework has the following distinguishable features: 1) the integration of localization into the communication system design; 2) the optimal RIS configuration to maximize the {conditional} achievable rate based on periodical location estimates; 3) the optimal precoding design to maximize the {conditional} achievable rate based on the estimation of
%\DD{better to emphasize: low-overhead instantaneous CSI ...?} 
instantaneous \ac{CSI}; 4) the reduction of the overhead for instantaneous \ac{CSI} acquisition with the help of localization. These targets are significantly different from existing two-timescale schemes \cite{HanTanJinWenMa19, HuDaiHanWan21,  ZhiPanRenWan21,Pan21} {because the localization of the UEs is integrated with the tasks of RIS optimization and channel estimation;} and because {we rely upon the statistical position-based RIS optimization scheme proposed in \cite{abrardo2020intelligent} by considering only imperfect instantaneous CSI}.
%---------------------------------------------

This paper is structured as follows. In Section \ref{sec_model}, the system models including the RIS reflection and channel models are introduced. The considered time scales (fast and slow time scales) and the proposed framework are presented in Section \ref{sec_Framework}. The location-coherent optimization method is presented in Section \ref{sec_LoccoherentOpt} and the channel-coherent optimization algorithm is described in Section \ref{sec_ChcoherentOpt}. The numerical results are presented in Section \ref{sec_results}, and conclusions are drawn in Section \ref{sec_conclusion}.
%---------------------------------------------
\subsubsection*{Notations}
Vectors and matrices are denoted by bold lowercase and uppercase letters,
respectively. The notations $(\cdot)^*$, $(\cdot)^{\mathrm{T}}$, $(\cdot)^{\mathrm{H}}$, $(\cdot)^{-1}$, and $(\cdot)^{\dagger}$, are reserved for the conjugate, transpose, conjugate transpose, inverse, and Moore-Penrose pseudoinverse operations. The expectation is denoted by $\mathrm{E}\{ \cdot \}$. The notation $\mathrm{Diag}(\boldsymbol{a})$ is to form a diagonal matrix with $\boldsymbol{a}$ being the diagonal elements. The operation $\mathrm{vec}(\boldsymbol{A})$ is to transform the matrix $\boldsymbol{A}$ into a column vector by stacking the columns on top of one another. The symbol $\otimes$ denotes the Kronecker product. $\delta\left(\cdot \right)$ is the Dirac delta function. $\Re \{ \cdot \}$ and $\Im \{ \cdot \}$ denote the real and imaginary part, respectively. The 2-argument arctangent function $\arctan2(x, y)$ returns a single value $\theta$ such that $-\pi <\theta \leq \pi$ and, for $r=\sqrt{x^2+y^2}$, $x=r\cos(\theta)$ and $y=r\sin(\theta)$. The inverse of the cosine function is denoted by $\arccos(\cdot)$. The Frobenius norm is denoted by $\Vert\cdot\Vert_{\mathrm{F}}$. The three-dimension (3D) rotation group of special orthogonal matrices is denoted by $\text{SO}(3)$. Complex Gaussian random vectors are denoted by $\boldsymbol{a} \sim \mathcal{CN}(\bar{\boldsymbol{a}}, \boldsymbol{R}_{\mathrm{a}})$ with $\mathrm{E}\{\boldsymbol{a} \} = \bar{\boldsymbol{a}}$, and $\mathrm{E}\{(\boldsymbol{a} - \bar{\boldsymbol{a}})(\boldsymbol{a} - \bar{\boldsymbol{a}})^{\mathrm{H}} \} )= \boldsymbol{R}_{\mathrm{a}}$; if all entries of $\boldsymbol{a}$ are real numbers, we use $\boldsymbol{a} \sim \mathcal{N}(\bar{\boldsymbol{a}}, \boldsymbol{R}_{\mathrm{a}})$.

\section{Communication System and Channel Model}\label{sec_model}
We consider a narrowband system with multiple RISs and multiple users as shown in Fig. \ref{fig:MultiRIS_MultiUser}. We assume that the \ac{LOS} paths from the \ac{BS} to the \acp{UE} are blocked, so that only the NLOS paths are present in the BS-UE direct links.\footnote{This is the deployment scenario where an RIS can be more suitable for improving the communication performance.} The carrier frequency is $f_0$, and the corresponding wavelength is $\lambda = c/f_0$ where $c$ is the speed of light. A uniform rectangular array (URA) is deployed at the BS whose center-location is $\boldsymbol{p}_{\mathrm{B}}\in \mathbb{R}^3$, with $N_\mathrm{B} = N_\mathrm{B,x} \times N_\mathrm{B, y}$ antenna array elements and the orientation matrix $\boldsymbol{O}_{\mathrm{B}}\in \text{SO}(3)$. In the system, there are $K$ RISs whose center-locations and orientations are denoted by $\boldsymbol{p}_{\mathrm{R}_k}\in \mathbb{R}
^3$ and $\boldsymbol{O}_{\mathrm{R}_k}\in \text{SO}(3)$, respectively. Each RIS comprises $P = N_{\mathrm{R}_k, \mathrm{x}} \times N_{\mathrm{R}_k, \mathrm{y}}$ unit cells (meta-atoms), %\DD{Note that in the introduction we called them meta-atoms} 
forming a URA. {Based on the knowledge of the BS location, we can either optimize the locations of the RISs or we can select the RISs in the network such that the BS-RIS links are in LOS.} The number of \Acp{UE} under simultaneous service is $I\ge 1$. The $i$-th UE is equipped with a URA of $N_{\mathrm{U}_i} = N_{\mathrm{U}_i, \mathrm{x}} \times N_{\mathrm{U}_i, \mathrm{y}}$ antenna elements, whose center-location is $\boldsymbol{p}_{\mathrm{U}_i}\in \mathbb{R}
^3$ and whose orientation matrix is $\boldsymbol{O}_{\mathrm{U}_i}\in \text{SO}(3)$.
All arrays have {cell/element} spacing equal to $\lambda/2$.\footnote{In RISs whose elements have inter-distances smaller than  half-wavelength, the mutual coupling needs to be taken into account \cite{RenZapDebAloYueRosTre20}. This generalization is postponed to a future research work.}
\begin{figure*}
    \centering
    \includegraphics[width=0.7\linewidth]{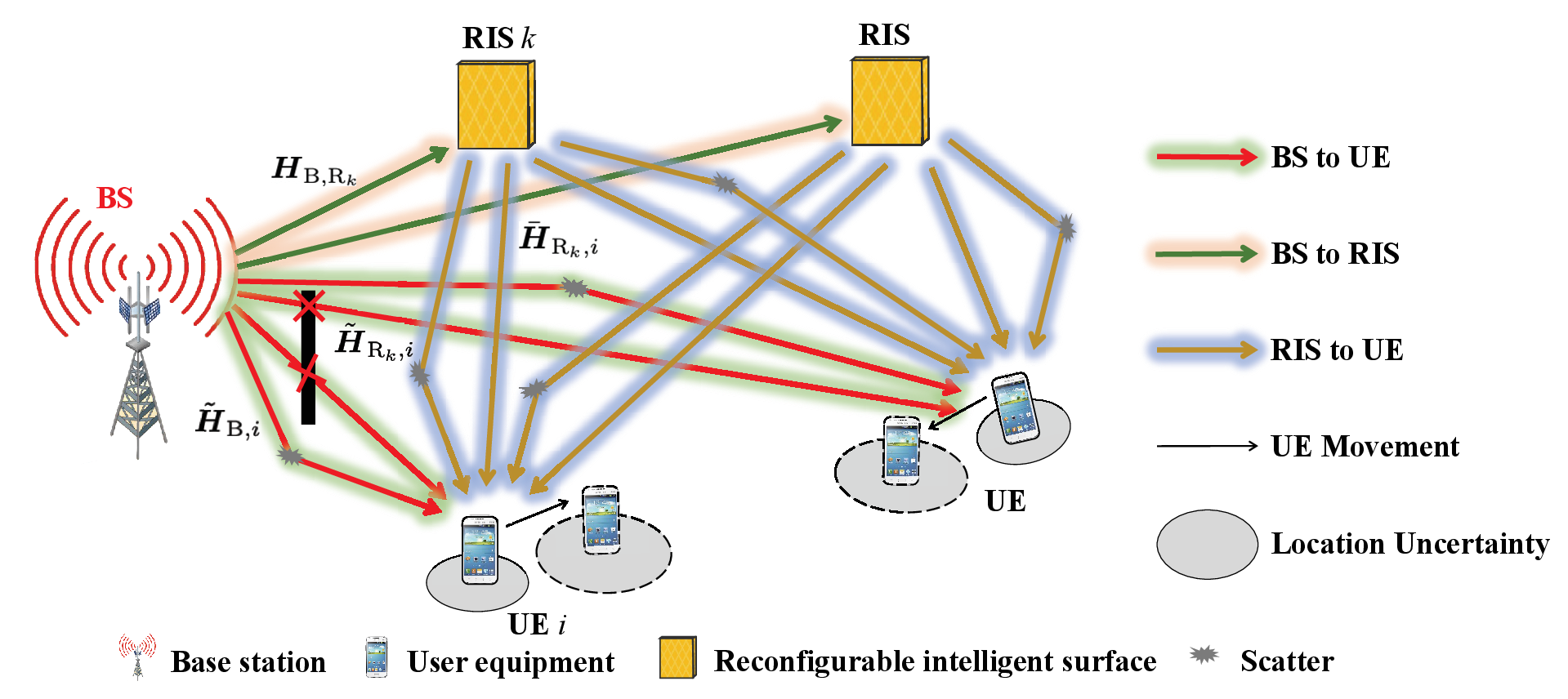}
    \caption{Application scenario: integrated localization and communication with multiple RISs and UEs. We assume that the LOS path from the BS to the UE is blocked while the RIS-BS links are in LOS. The locations of the UEs are not perfectly known and the UEs can randomly move throughout the network.}\label{fig:MultiRIS_MultiUser} 
\end{figure*}

%\vspace{-8mm}
\subsection{Reflection Model for the RIS}
Assuming each unit cell is sufficiently small to be considered in the far-field region of the \ac{BS} and {UE}, the local reflection coefficient of the $p$-th unit cell towards the general direction of scattering $\boldsymbol{\Theta} = (\theta_{\mathrm{az}}, \theta_{\mathrm{el}})$, {with $\theta_{\mathrm{az}}$ and $\phi_{\mathrm{el}}$} denoting the azimuth and elevation angles in the local coordinate system of the RIS, can be modeled as $ r_p \left( \boldsymbol{\Theta}_{\mathrm{inc}}, \boldsymbol{\Theta} \right) = \sqrt{F\left( \boldsymbol{\Theta}_{\mathrm{inc}} \right) F\left( \boldsymbol{\Theta} \right)} G_\mathrm{c} b_p$ \cite{abrardo2020intelligent}, where $F ( \boldsymbol{\Theta} )$ is the normalized power radiation pattern of each unit cell, which is assumed to be frequency-independent within the bandwidth of interest, and is defined as
\begin{align}
F\left( \boldsymbol{\Theta} \right) = \left\{ 
\begin{array}{ll}
{\cos^q \left( \theta_{\mathrm{el}} \right),} & {\theta_{\mathrm{el}} \in \left[0, \pi/2 \right], \theta_{\mathrm{az}} \in \left[0, 2\pi \right]}; \\
0, & \text{otherwise}
\end{array}
\right. \label{eqnprp}
\end{align}
where $q$ is a tunable parameter, $\boldsymbol{\Theta}_{\mathrm{inc}} = (\theta_{\mathrm{az}, \mathrm{inc}}, \theta_{\mathrm{el}, \mathrm{inc}})$ is the angle of incidence with respect to the RIS, and $G_\mathrm{c}$ is the boresight gain of the unit cell \cite{Ell21,abrardo2020intelligent, DegVitDiTre22}. The term $b_p$, which is often referred to as the load reflection coefficient, is defined as $b_p = \rho_p \exp (\jmath \phi_p)$, where $\rho_p$ and $\phi_p$ are the amplitude and phase of the $p$-th RIS element, respectively. The load reflection coefficient of each unit cell is the parameter of the RIS that can be optimized for performance improvement. In this paper, for simplicity, we assume $\rho_p=1$. Therefore, the optimization of $\boldsymbol{b} = [b_1, b_2, \cdots, b_P]^{\mathrm{T}}$ is equivalent to optimize the RIS phase profile. 

\subsection{Indirect Link Models (BS-RIS-UE Channel)}
As mentioned, the BS, RISs, and UEs are equipped with URAs with $N_{\mathrm{x}} \times N_{\mathrm{y}}$ elements. For a given direction $\boldsymbol{u} = [\sin (\theta_{\mathrm{el}}) \cos (\theta_{\mathrm{az}}), \sin (\theta_{\mathrm{el}}) \sin (\theta_{\mathrm{az}}), \cos (\theta_{\mathrm{el}})]^{\mathrm{T}}$, where $\theta_{\mathrm{el}}$ and $\theta_{\mathrm{az}}$ are the elevation and azimuth angles, respectively, the steering vector is defined as $\boldsymbol{a}( \boldsymbol{Q}, \boldsymbol{\theta} ) = [a_{1}, a_{2}, \cdots, a_{N_{\mathrm{x}} N_{\mathrm{y}}}]^{\mathrm{T}}$, with the $((n_{\mathrm{x}} - 1)\times N_{\mathrm{y}} + n_{\mathrm{y}})$-th element given by $ a_{(n_{\mathrm{x}} - 1)\times N_{\mathrm{y}} + n_{\mathrm{y}}} = \exp ( {\jmath 2\pi   }/{\lambda}\boldsymbol{q}_{n_{\mathrm{x}}, n_{\mathrm{y}}}^{\mathrm{T}} \boldsymbol{u} )$ and $\mathbf{q}_{n_{\mathrm{x}}, n_{\mathrm{y}}} = [ ( n_{\mathrm{x}} - 1 ) \lambda/2 - ( N_{\mathrm{x}} - 1 ) \lambda/4, ( n_{\mathrm{y}} - 1 ) \lambda/2 - ( N_{\mathrm{y}} - 1 ) \lambda/4, 0 ]^{\mathrm{T}}$ is the position of the $(n_{\mathrm{x}}, n_{\mathrm{y}})$-th antenna element for $n_{\mathrm{x}}=1, 2, \cdots, N_{\mathrm{x}}$ and $n_{\mathrm{y}}=1, 2, \cdots, N_{\mathrm{y}}$. We use the notation $\boldsymbol{\theta} = [\theta_{\mathrm{el}}, \theta_{\mathrm{az}}]^{\mathrm{T}}$ for the angles and $\boldsymbol{Q} \in \mathbb{R}^{N_{\mathrm{x}} N_{\mathrm{y}} \times 3 }$ with $\boldsymbol{Q} = [\boldsymbol{q}_{1, 1}, \boldsymbol{q}_{1, 2}, \cdots, \boldsymbol{q}_{N_{\mathrm{x}}, N_{\mathrm{y}}} ]^{\mathrm{T}}$ for the positions of the antenna elements in the array's local coordinate system. 

\subsubsection{BS-RIS channel} 
The channel matrix from the BS to the $k$-th RIS, $\boldsymbol{H}_{\mathrm{B}, \mathrm{R}_k} \in \mathbb{C}^{P\times N_{\mathrm{B}}}$, is
\begin{align}
\boldsymbol{H}_{\mathrm{B}, \mathrm{R}_k} = \alpha_{\mathrm{B}, k} \boldsymbol{a} \left(\boldsymbol{Q}_{\mathrm{R}_k}, \boldsymbol{\phi}_{\mathrm{R}_k} \right) \boldsymbol{a}^{\mathrm{T}} \left( \boldsymbol{Q}_{\mathrm{B}}, \boldsymbol{\theta}_{\mathrm{B}, \mathrm{R}_k} \right),
\end{align}
and $\alpha_{\mathrm{B}, k}={\sqrt{F(\boldsymbol{\phi}_{\mathrm{R}_k})G_{\mathrm{T}} G_{\mathrm{c}}} \lambda }/({4\pi | \boldsymbol{p}_{k, \mathrm{B}} | }) \exp (- \jmath {2\pi | \boldsymbol{p}_{k, \mathrm{B}} |}/{\lambda} )$ is the path gain of the LOS from the BS to the $k$-th RIS. The matrices
$\boldsymbol{Q}_{\mathrm{B}} \in \mathrm{R}^{N_{\mathrm{B}, \mathrm{x}} N_{\mathrm{B}, \mathrm{y}} \times 3}$ and $\boldsymbol{Q}_{\mathrm{R}_k} \in \mathrm{R}^{P \times 3}$ contain the positions of the antenna elements of the BS and the $k$-th RIS, respectively, and $\boldsymbol{\theta}_{\mathrm{B}, \mathrm{R}_k} = [\theta_{\mathrm{B}, k}^{\mathrm{el}}, \theta_{\mathrm{B}, k}^{\mathrm{az}} ]^{\mathrm{T}}$
is the AOD from the BS to the $k$-th RIS, which corresponds to the direction of the vector $\boldsymbol{p}_{k, \mathrm{B}} = \boldsymbol{O}_{\mathrm{B}} (\boldsymbol{p}_{\mathrm{R}_k} - \boldsymbol{p}_{\mathrm{B}})$ in the local coordinate system of the BS.
The elevation and azimuth angles are given by $\theta_{\mathrm{B}, k}^{\mathrm{az}} =\arctan2([\boldsymbol{p}_{k, \mathrm{B}}]_{2},[\boldsymbol{p}_{k, \mathrm{B}}]_{1})$ and $\theta_{\mathrm{B}, k}^{\mathrm{el}} =\arccos\left([\boldsymbol{p}_{k, \mathrm{B}}]_{3}/\Vert\boldsymbol{p}_{k, \mathrm{B}}\Vert\right)$,
respectively. Similarly, $\boldsymbol{\phi}_{\mathrm{R}_k} = [\phi_{\mathrm{R}_k}^{\mathrm{el}}, \phi_{\mathrm{R}_k}^{\mathrm{az}}]$ is the  AOA at the $k$-th RIS from the BS.

\subsubsection{RIS-UE channel} The channel from the $k$-th RIS to the $i$-th UE is denoted by $\boldsymbol{H}_{\mathrm{R}_k, i} = \bar{\boldsymbol{H}}_{\mathrm{R}_k, i} (\boldsymbol{p}_{\mathrm{U}_i}) + \tilde{\boldsymbol{H}}_{\mathrm{R}_k, i}$, and it is composed of a LOS component (denoted by $\bar{\boldsymbol{H}}_{\mathrm{R}_k, i} (\boldsymbol{p}_{\mathrm{U}_i}) \in \mathbb{C}^{N_{\mathrm{U}_i} \times P}$) and an NLOS component (denoted by $\tilde{\boldsymbol{H}}_{\mathrm{R}_k, i} \in \mathbb{C}^{N_{\mathrm{U}_i} \times P}$). Based on a Rice channel model, and defining the Rician factor $\kappa_{k, i}\ge 0$, the LOS channel component is given by \cite{Pan21}
\begin{align}
\bar{\boldsymbol{H}}_{\mathrm{R}_k, i} (\boldsymbol{p}_{\mathrm{U}_i}) = \bar{\alpha}_{k, i} \boldsymbol{a} \left(\boldsymbol{Q}_{\mathrm{U}_i}, \boldsymbol{\phi}_{\mathrm{R}_k, \mathrm{U}_i} \right) \boldsymbol{a}^{\mathrm{T}} \left( \boldsymbol{Q}_{\mathrm{R}_k}, \boldsymbol{\theta}_{\mathrm{R}_k, \mathrm{U}_i} \right),
\end{align}
where $\bar{\alpha}_{k, i} = \sqrt{{\kappa_{k, i} \rho_{k, i}}/{(\kappa_{k, i} + 1)}} \exp (- \jmath {2\pi | \boldsymbol{p}_{{\mathrm{R}_k}, {\mathrm{U}_i}} |}/{\lambda} )$ is the path gain of the LOS from the $k$-th RIS to the $i$-th UE
%, given by
with $\rho_{k, i} = {F(\boldsymbol{\theta}_{\mathrm{R}_k, \mathrm{U}_i})G_{\mathrm{R}} G_{\mathrm{c}} \lambda^2}/{(16\pi^2 \| \boldsymbol{p}_{{\mathrm{R}_k}, {\mathrm{U}_i}} \|^{\alpha})}$,
where $\boldsymbol{p}_{{\mathrm{R}_k}, {\mathrm{U}_i}} = \boldsymbol{O}_{\mathrm{U}_i} (\boldsymbol{p}_{\mathrm{U}_i} - \boldsymbol{p}_{\mathrm{R}_k})$, and $\alpha$ is the path-loss exponent. The matrix $\boldsymbol{Q}_{\mathrm{U}_i} \in \mathbb{C}^{N_{\mathrm{U}_i} \times 3}$ contains the positions of the antenna elements %in the local coordinate system 
of the $i$-th UE. $\boldsymbol{\theta}_{\mathrm{R}_k, \mathrm{U}_i} = [{\theta}_{\mathrm{R}_k, \mathrm{U}_i}^{\mathrm{el}}, {\theta}_{\mathrm{R}_k, \mathrm{U}_i}^{\mathrm{az}}]^{\mathrm{T}}$ is the AOD from the $k$-th RIS to the $i$-th UE,
and $\boldsymbol{\phi}_{\mathrm{R}_k, \mathrm{U}_i} = [{\phi}_{\mathrm{R}_k, \mathrm{U}_i}^{\mathrm{el}}, {\phi}_{\mathrm{R}_k, \mathrm{U}_i}^{\mathrm{az}} ]^{\mathrm{T}}$ is the AOA from the $k$-th RIS to the $i$-th UE.

The NLOS component $\tilde{\boldsymbol{H}}_{\mathrm{R}_k, i}$ is modeled as a random matrix, and $\mathrm{vec}(\tilde{\boldsymbol{H}}_{\mathrm{R}_k, i}) \sim \mathcal{CN}(0, \tilde{\boldsymbol{R}}_{k,i})$.
When $\tilde{\boldsymbol{R}}_{k,i}$ is the identity matrix, the RIS-UE channel model reduces to the independent and identically distributed Rician fading model \cite{ZhiPanRenWan21,Pan21,DaiZhuPanRenWan22}.

\subsubsection{BS-RIS-UE channel}
Given the RIS profile $\boldsymbol{b}_k$, the overall channel matrix from the BS to the $i$-th UE via the $k$-th RIS, $\boldsymbol{H}_{k, i} (\boldsymbol{b}_k, \boldsymbol{p}_{\mathrm{U}_i}) \in \mathbb{C}^{N_{\mathrm{U_i}} \times N_{\mathrm{B}}}$, is given by
\begin{align}
\boldsymbol{H}_{k, i}(\boldsymbol{b}_k, \boldsymbol{p}_{\mathrm{U}_i}) = &(\bar{\boldsymbol{H}}_{\mathrm{R}_k, i}(\boldsymbol{p}_{\mathrm{U}_i}) + \tilde{\boldsymbol{H}}_{\mathrm{R}_k, i}) \mathrm{Diag} \left( \boldsymbol{b}_k \right) \boldsymbol{H}_{\mathrm{B}, \mathrm{R}_k}. \label{eq:CHBRU}
\end{align}
It can be shown that $\boldsymbol{h}_{k, i} (\boldsymbol{b}_k, \boldsymbol{p}_{\mathrm{U}_i}) = \mathrm{vec}(\boldsymbol{H}_{k, i} (\boldsymbol{b}_k, \boldsymbol{p}_{\mathrm{U}_i})) \sim \mathcal{CN}(\bar{\boldsymbol{h}}_{k, i} (\boldsymbol{b}_k, \boldsymbol{p}_{\mathrm{U}_i}), \boldsymbol{R}_{k, i}(\boldsymbol{b}_k) )$, with 
\begin{align}
\bar{\boldsymbol{h}}_{k, i} \left(\boldsymbol{b}_k, \boldsymbol{p}_{\mathrm{U}_i}\right) = &\big( \left( \mathrm{Diag} \left( \boldsymbol{b}_k \right) \boldsymbol{H}_{\mathrm{B}, \mathrm{R}_k} \right)^{\mathrm{T}} \otimes \boldsymbol{I}_{N_{\mathrm{U}_i}} \big) \notag \\
& \cdot \mathrm{vec} \left( \bar{\boldsymbol{H}}_{\mathrm{R}_k, i} (\boldsymbol{p}_{\mathrm{U}_i}) \right), \label{eq:meanCHBRU} \\
\boldsymbol{R}_{k, i} \left(\boldsymbol{b}_k \right) = & \left( \left( \mathrm{Diag} ( \boldsymbol{b}_k ) \boldsymbol{H}_{\mathrm{B}, k} \right)^{\mathrm{T}} \otimes \boldsymbol{I}_{N_{\mathrm{U}_i}} \right) \notag \\
& \cdot \tilde{\boldsymbol{R}}_{k,i} 
 \left( \left( \mathrm{Diag} ( \boldsymbol{b}_k ) \boldsymbol{H}_{\mathrm{B}, k} \right)^{\mathrm{T}} \otimes \boldsymbol{I}_{N_{\mathrm{U}_i}} \right)^{\mathrm{H}}.\label{eq:varCHBRU}
\end{align}

\begin{figure*}
    \centering
    \includegraphics[width=0.7\linewidth]{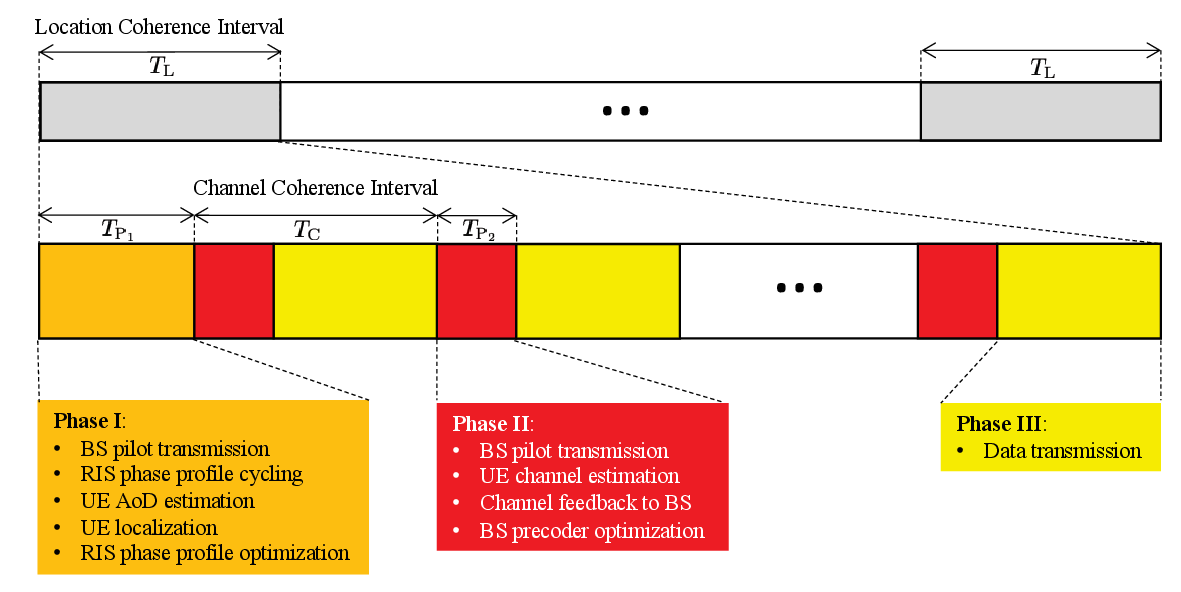}
    \caption{The proposed frame structure: we assume that the UE locations change slowly during the location coherence interval $T_{\mathrm{L}}$, and that the channel remains unchanged during the channel coherence interval $T_{\mathrm{C}}$. Each location coherence interval includes $N_\mathrm{C}$ channel coherence intervals, and it is split into three phases: 1) Phase I with duration $T_{\mathrm{P}_1}$, 2) Phase II with total duration $N_\mathrm{C}T_{\mathrm{P}_2}$, and 3) Phase III with total duration $N_\mathrm{C}(T_{\mathrm{C}} - T_{\mathrm{P}_2})$.} \label{fig:framestructure}
\end{figure*}
\subsection{Direct Link Model (BS-UE Channel)}
The NLOS component of the direct link consists of a number of clustered paths, each corresponding to a micro-level scattering path. As a result, the channel matrix of the direct link from the BS to $i$-th UE, which is denoted by $\tilde{\boldsymbol{H}}_{\mathrm{B}, i} \in \mathbb{C}^{N_{\mathrm{U}_i} \times N_\mathrm{B}}$, can be modeled as a random matrix, where its elements are i.i.d. $\mathcal{CN}(0, \sigma_{\mathrm{B}, i}^2)$ random variables.

\section{Separation of Time Scales and Proposed Framework}\label{sec_Framework}
In the considered integrated localization and communication framework, we consider several location coherence intervals $T_{\mathrm{L}}$ and divide such a period into three phases, as shown in Fig. \ref{fig:framestructure}. Note that the channel changes faster than the locations. Therefore, it is reasonable to assume that the location coherence interval includes multiple channel coherence intervals.

\subsection{Time Scales}
To account for user mobility and channel variations, we consider two time indices, $t$ and $\tau$. 

\subsubsection{Fast time scale}
The fast time index $t$ corresponds to the individual narrowband signals transmitted within one channel coherence interval, whose duration is $T_{\text{C}}$ and is proportional to $\lambda/v$, where $v$ is the maximum \ac{UE} speed.  Hence, the number of transmissions within one channel coherence interval is $T_{\text{C}}B$, where $B$ is the signal bandwidth. 
We denote by $\boldsymbol{b}_{k, t} = [b_1^{(k, t)}, b_2^{(k, t)}, \cdots, b_P^{(k, t)}]^{\mathrm{T}}$ the load reflection coefficients of the $k$-th RIS at time instant $t$. The transmit signal (precoded data or pilot) is denoted by  $\boldsymbol{x}_t = [x_{t, 1}, x_{t, 2}, \cdots, x_{t, {N_\mathrm{B}}}]^{\mathrm{T}}$. The received signal at the $i$-th UE, i.e., $\boldsymbol{y}_{i, t} = [y_1^{(i, t)}, y_2^{(i, t)}, \cdots, y_{N_{\mathrm{U}}}^{(i, t)}]^{\mathrm{T}}$, can be written as
\begin{align}
\boldsymbol{y}_{i, t} = & \boldsymbol{H}_{\mathrm{B}, i} (\mathcal{B}_t, \boldsymbol{p}_{\mathrm{U}_i}) \boldsymbol{x}_t + \boldsymbol{n}_{t}, \label{eqRec}
\end{align}
where $\boldsymbol{H}_{\mathrm{B}, i} (\mathcal{B}_t, \boldsymbol{p}_{\mathrm{U}_i}) = \tilde{\boldsymbol{H}}_{\mathrm{B}, i} + \sum_{k=1}^K \boldsymbol{H}_{k, i} (\boldsymbol{b}_{k, t}, \boldsymbol{p}_{\mathrm{U}_i})$ is the cascaded channel from the BS to UE $i$, with $\mathcal{B}_t = \{\boldsymbol{b}_{1, t}, \boldsymbol{b}_{2, t}, \cdots, \boldsymbol{b}_{K, t} \}$. The white Gaussian noise (AWGN) is denoted by $\boldsymbol{n}_{t}\sim \mathcal{CN} (0, \sigma_n^2 \boldsymbol{I}_{N_\mathrm{U}})$ with $\sigma_n^2 = n_{\mathrm{f}} N_0 B$, where $n_{\mathrm{f}}$ is the noise factor, and $N_0$ is the noise power density.

\subsubsection{Slow time scale}\label{sec_UE_movement}
In contrast, $\tau \in \{0, 1, \cdots, N_{\mathrm{C} }- 1 \}$ is the time index related to the time scale at which the user moves, with sampling time $T_{\text{C}}$. $N_{\mathrm{C}}$ is the number of channel coherence intervals in one location coherence interval (after $T_{\text{P}_1}$), i.e., $T_{\text{L}} = T_{\text{P}_1} + N_{\mathrm{C}}T_{\text{C}}$.
We model the user movement by a random walk process, where the user position at the beginning of a frame is denoted by $\boldsymbol{p}_{\text{U}_i,0}$, while during the coherence interval $\tau$, it is $\boldsymbol{p}_{\text{U}_i,\tau}=\boldsymbol{p}_{\text{U}_i,\tau-1} + \boldsymbol{v}_{i,\tau}$, where $\boldsymbol{v}_{i,\tau} \sim \mathcal{N}(\boldsymbol{0},\boldsymbol{\Sigma}_{\text{pos}})$ and $\boldsymbol{\Sigma}_{\text{pos}}$ is the covariance of the random walk with sampling period $T_{\text{C}}$. 

\subsection{Proposed Frame Structure}\label{Loc_est_model}
\subsubsection{Phase I}
Phase I is intended to estimate the UE locations as well as to optimize the RIS phase profile. To this end, the BS transmits the pilot sequence $\boldsymbol{s}_t^{\mathrm{P}_1}$, $t=1, 2, \cdots, T_{\mathrm{P}_1}B$, and the RISs configures the phase profile $\boldsymbol{b}_{k, t}^{\mathrm{P}_1}$. With the received sequences, we estimate the UE locations based on AOD estimations (see Section \ref{sec_UELocEst}). In addition, we optimize the phase profile of all RISs based on the estimates of the locations and on the covariance matrix of the location uncertainty (see Section \ref{sec_RISPhase}). The optimized phase profiles remain unchanged during Phases II and III until the UE locations are outdated. Phase I is repeated every location coherence interval in order to appropriately update the RIS phase profiles. 

\subsubsection{Phase II}
We divide the remaining time $T_{\mathrm{L}} - T_{\mathrm{P}_1}$ into several channel coherent blocks and assume that the channel remains unchanged in each coherent block. Phase II in each channel coherent interval $T_{\mathrm{C}}$ is intended to obtain the estimate of the composed channel, based on which the optimal precoders of the BS are designed to maximize the achievable rate. To this end, the BS transmits the pilot sequence $\boldsymbol{s}_t^{\mathrm{P}_2}$, $t=1, 2, \cdots, T_{\mathrm{P}_2}B$, and channel estimation is performed on the received pilot symbols (see Section \ref{sec_CHEst}), without involving any configuration of the RIS.\footnote{When a time-division duplexing system is employed, the downlink CSI can be estimated by taking advantage of channel reciprocity through the transmission of pilots in the uplink. Therefore, the proposed framework can still be utilized, by applying an appropriate processing to the channel estimates and to the covariance matrix of the channel estimation errors. However, in this work, we limit our analysis to the downlink only, where the CSI at the BS is obtained through channel feedback \cite{joung2016channel, ma2021model}.}

\subsubsection{Phase III}
The Phase III is intended to perform data transmission with the  RIS phase profile  optimized in Phase I and with the optimal BS precoder optimized in Phase II.

Thus, the localization accuracy of the UEs impacts the communication performance. The user movement, the NLOS components, and the pilots used in Phase I impact the localization accuracy of the UEs, and the channel estimation accuracy in Phase II impacts the achievable rate. Our aim is to understand how localization helps channel estimation and data communication.

\section{Location-coherent Optimization Phase}\label{sec_LoccoherentOpt}
Here, we describe the operations in Phase I, where the BS transmits $T_{\text{P}_1}B$ pilots to estimate the \ac{UE} positions. From the estimated positions, we describe how the \ac{RIS} phase profile is optimized. For both positioning and phase profile optimization, we use performance bounds, in order to obtain fundamental performance insights and to be agnostic to specific algorithms for localization. 

\subsection{UE Location Estimation Bound}\label{sec_UELocEst}
We use Fisher information analysis \cite{KeyKesGraWym20} to obtain the covariance matrix $\boldsymbol{\Sigma}_{i,0}$ of the position of the $i$-th user given the observations $\boldsymbol{y}_{i,t}$ at the beginning of a frame, corresponding to $\tau=0$. For simplicity, but with no loss of generality, we assume that the same pilot symbols are transmitted at each time instance, i.e., $\boldsymbol{x}_t=\boldsymbol{x}$ for all $t$. As the BS-UE channel does not convey location information, it is treated as an interference term in the received signal. Therefore, we simply remove the received signal from the direct link (if any) by designing the RIS phases such that $\boldsymbol{b}_{k, 2\check{t}-1}^{\mathrm{P}_1}=-\boldsymbol{b}_{k, 2\check{t}}^{\mathrm{P}_1}$, for $\check{t}=1, 2, \cdots, T_{\mathrm{P}_1/2}B$, (we assume that $T_{\mathrm{P}_1}B$ is an even number).\footnote{The considered design of the RIS phase profile does not require either prior knowledge of the UE locations or signaling overhead from the BS. Further details can be found in \cite{keykhosravi2021multi}.} Next, at the receiver, we calculate $\check{\boldsymbol{y}}_{i, t}$ as
\begin{align}
\check{\boldsymbol{y}}_{i, \check{t}} = & \frac{1}{2}(\boldsymbol{y}_{i, 2\check{t}-1}-\boldsymbol{y}_{i, 2\check{t}}) \notag \\
= & \sum_{k=1}^K \bar{\boldsymbol{H}}_{\mathrm{R}_k, i} \mathrm{Diag} ( \check{\boldsymbol{b}}_{k,\check{t}} ) \boldsymbol{H}_{\mathrm{B}, \mathrm{R}_k} \boldsymbol{x} + \check{\boldsymbol{n}}_{\check{t},i}. \label{eqRec2}
\end{align}
Here, \eqref{eqRec2} follows by substituting \eqref{eq:CHBRU} into \eqref{eqRec}. Also $\check{\boldsymbol{b}}_{k,\check{t}} = {\boldsymbol{b}}_{k,2\check{t}}$, and  $\check{\boldsymbol{n}}_{\check{t},i}\sim \mathcal{CN}(\boldsymbol{0},\check{\sigma}^2_n\boldsymbol{I}_{N_{\mathrm{U}_i}})$ models the AWGN and the RIS-UE multipath, where $\check{\sigma}^2_n=\sigma_n^2/2+\sigma_i^2$ and $\sigma_i^2=\sum_{k=1}^K  \Vert \boldsymbol{H}_{\mathrm{B}, \mathrm{R}_k}\boldsymbol{x}\Vert_{\mathrm{F}}^2 \sigma^2_{k,i}$. The NLOS components in \eqref{eqRec2} have opposite sign and are added coherently. Also, the NLOS components for different values of $\check{t}$ can be modeled as independent from one another, since the RIS phase profile is configured randomly and independently for different $\check{t}$. In addition, the employed configuration for the phase profile does not require prior knowledge of the UE locations, and it is computationally efficient.

To compute $\boldsymbol{\Sigma}_{i,0}$ from \eqref{eqRec2}, we use a two-step approach. We  define the channel parameter vector for the $i$-th user $\boldsymbol{\eta}_i\in \mathbb{R}^{6K}$ as  $\boldsymbol{\eta}_i = [ \boldsymbol{\theta}_{i}^{\mathrm{T}},\boldsymbol{\phi}_{i}^{\mathrm{T}}, \boldsymbol{\bar{\alpha}}_{ i}    ]^{\mathrm{T}}$, where $\boldsymbol{\theta}_{i}\in \mathbb{R}^{2K}$ and $\boldsymbol{\phi}_{i}\in \mathbb{R}^{2K}$ are vectors containing the AODs  (at the RISs) and AOAs (at the UEs) from the RISs to the $i$-th UE, respectively, and   $\boldsymbol{\bar{\alpha}}_{ i}  $ contains the channel gains. Specifically, $\boldsymbol{\theta}_{i}=[ \boldsymbol{\theta}_{\mathrm{R}_1, \mathrm{U}_i}^{\mathrm{T}},\dots, \boldsymbol{\theta}_{\mathrm{R}_K, \mathrm{U}_i}^{\mathrm{T}}]^{\mathrm{T}}$, $\boldsymbol{\phi}_{i}=[ \boldsymbol{\phi}_{\mathrm{R}_1, \mathrm{U}_i}^{\mathrm{T}},\dots, \boldsymbol{\phi}_{\mathrm{R}_K, \mathrm{U}_i}^{\mathrm{T}}]^{\mathrm{T}}$, and  $\boldsymbol{\bar{\alpha}}_{ i} = [\Re(\bar{\alpha}_{1, i}),\Im(\bar{\alpha}_{1, i}), \dots,\Re(\bar{\alpha}_{K, i}),\Im(\bar{\alpha}_{K, i})]^{\mathrm{T}}$. Next, we compute the Fisher information matrix (FIM) for the channel parameters, which is denoted by $\boldsymbol{J}(\boldsymbol{\eta}_i)\in \mathrm{R}^{6K\times 6K}$, as $\boldsymbol{J}(\boldsymbol{\eta}_i)=\frac{2}{\check{\sigma}_i^{2}}\sum_{\check{t}=1}^{T_{\mathrm{P}_1}B/2}\Re\big\{ \nabla_{\boldsymbol{\eta}_i}\boldsymbol{\mu}_{i,\check{t}}(\nabla_{\boldsymbol{\eta}_i}\boldsymbol{\mu}_{i,\check{t}})^{\mathrm{H}}\big\}$ \cite{Kay:93}. Here, $ \boldsymbol{\mu}_{i,\check{t}}$ is the noise-free observation of \eqref{eqRec2} given by $\boldsymbol{\mu}_{i,\check{t}}=\sum_{k=1}^K \bar{\boldsymbol{H}}_{\mathrm{R}_k, i} \mathrm{Diag} \left( \check{\boldsymbol{b}}_{k,\check{t}} \right) \boldsymbol{H}_{\mathrm{B}, \mathrm{R}_k} \boldsymbol{x}$, where $\nabla_{\bm{\eta}}\boldsymbol{\mu}_{i,t} \in \mathbb{C}^{6K\times N_{\mathrm{U}_i}}$ is obtained as $[\nabla_{\bm{\eta}}\boldsymbol{\mu}_{i,t}]_{r,s} = {\partial [\boldsymbol{\mu}_{i,t}]_{s} }/{\partial [\boldsymbol{\eta}]_r}$, which can be calculated based on the relations in Section~\ref{sec_model}. Also, we define the vector of position parameters as $\boldsymbol{\zeta}_i\in\mathbb{R}^{(4K+3)}$, where $   \boldsymbol{\zeta}_i  = [ \boldsymbol{p}_{\mathrm{U}_{i,0}}^{\mathrm{T}},\boldsymbol{\phi}_{i}^{\mathrm{T}}, \boldsymbol{\bar{\alpha}}_{ i}    ]^{\mathrm{T}}$.\footnote{We localize the UE using the AODs only.} The corresponding FIM matrix can be calculated as $\boldsymbol{J}(\boldsymbol{\zeta}_i)= \boldsymbol{\varUpsilon}_i^{\mathrm{T}} \boldsymbol{J}(\boldsymbol{\eta}_i) \boldsymbol{\varUpsilon}_i$. Here, $\boldsymbol{\varUpsilon_i} \in \mathbb{R}^{6K\times (4K+3)}$ is the Jacobian matrix defined as $[\boldsymbol{\varUpsilon}_i]_{s,r} = {\partial [\boldsymbol{\eta}_i]_{s} }/{\partial [\boldsymbol{\zeta}_i]_r}$, which can be calculated based on the geometrical relations described in Section\,\ref{sec_model}.  The localization  error matrix ${\boldsymbol{\Sigma}}_{i,0}\in\mathbb{R}^{3\times3}$
is obtained as the first $3\times3$ diagonal block of the inverse
of $\boldsymbol{J}({\boldsymbol{\zeta}}_i)$, i.e., 
\begin{align}
{\boldsymbol{\Sigma}}_{i,0}=\left[\boldsymbol{J}^{-1}(\boldsymbol{\zeta}_i)\right]_{1:3,1:3}. \label{PosPhase1}
\end{align}

For comparison, it is interesting to analyze the case study when Phase I is not introduced in Fig. \ref{fig:framestructure}. In this case, the UE location estimates can be obtained using conventional positioning techniques, e.g., fingerprinting and GNSS-based methods \cite{IndoorLoc16}. For simplicity, we refer to this scenario as the \emph{prior}-based scheme with the covariance matrix of location uncertainty given by
\begin{align}
{\boldsymbol{\Sigma}}_{i,0}={\boldsymbol{\Sigma}}_{i,0}^{\mathrm{pri}}. \label{PosPri}
\end{align}

In each case, we can generate a synthetic estimate of the position as $ \hat{\boldsymbol{p}}_{\mathrm{U}_{i,0}} = \boldsymbol{p}_{\mathrm{U}_{i,0}} + \boldsymbol{w}_i$, where $\boldsymbol{w}_i$ is a realization of a random vector with distribution $\mathcal{N}(\boldsymbol{0},{\boldsymbol{\Sigma}}_{i,0})$.

In  Section\,\ref{sec_results}, we show that a practical estimation algorithm achieves the proposed bounds for sufficiently large values of the SNR. We anticipate that we consider a scenario with two RISs, in which the considered channel estimator operates as follows. First, it separates the received signal from each RIS by using a temporal code of a given length, as described in \cite{keykhosravi2021multi}. Then, for each RIS-aided path, it estimates a one-dimensional AoA (assuming the UE is equipped with a linear-array antenna). Based on the estimated AoA, the signals received at different UE antennas are added constructively. Subsequently, the AODs are estimated by using the algorithm in \cite[Sec.\,IV-C]{KeyKesGraWym20}. Finally, a coarse estimate of the UE position is obtained by finding the intersection\footnote{If the two lines do not cross, a point with minimum aggregated distance from the two lines is selected.} of the two lines corresponding to the two AODs. The obtained estimate is then refined by maximizing the  log-likelihood function by using the coarse estimate as the initial point.

It is worth noting that the localization performance could be improved if LOS paths exist between the BS and UE. This case study is postponed to a future research work, since the performance gains depend on the relative strengths between the RIS-aided and LOS paths.

\subsection{RIS Phase Profile Optimization}\label{sec_RISPhase}
The RIS phase profile is optimized based on the statistics of the terminal locations and the NLOS components of all associated connections.
To this end, we take inspiration from the location-based RIS optimization approach proposed in \cite{abrardo2020intelligent}, which assumes that the actual locations of the UEs are uniformly distributed around an estimated location. How the localization of the UEs is performed is, however, not specified in \cite{abrardo2020intelligent}. Therein, the localization statistics are assumed to be available from external sources, e.g., based on the GNSS. By contrast, we consider the localization scheme in Section \ref{sec_UELocEst}, i.e., the RIS optimizer assumes that the actual position of the UEs is distributed around the estimated position according to a normal distribution with covariance matrix given by \eqref{PosPri}.

To elaborate, let $\mathcal{P}$ be a sample set of random variables with distributions $\mathcal{N}\left(\hat{\boldsymbol{p}}_{\mathrm{U}_{i,0}},\boldsymbol{\Sigma}_{i,0}\right)$, where $\hat{\boldsymbol{p}}_{\mathrm{U}_{i,0}}$ is the estimated UE position. Similarly, let $\mathcal{T}$ be a sample set of realizations of the NLOS matrices $\tilde{\boldsymbol{H}}_{\mathrm{B},i}$ and
$\tilde{\boldsymbol{H}}_{\mathrm{R}_k,i}$ generated according to the corresponding statistics and introduce $\Omega = \left\{\mathcal{P},\mathcal{T}\right\}$. Therefore, we denote by $\mathcal{B}=\left\{\boldsymbol{b}_{1},\boldsymbol{b}_{2},\dots,\boldsymbol{b}_{K}\right\}$ the set of vectors containing the reflection coefficients of the $K$ RISs. Note that for any realization ${\boldsymbol{\omega}} \in \Omega$ and for any $\mathcal{B}$, all channel matrices in \eqref{eqRec} can be determined. 
The system thus turns out to be a classical MIMO communication system whose optimal behavior can be determined according to established methods.
Based on \cite{abrardo2020intelligent}, it is therefore possible to optimize the conditional sum rate $R_{\text{tot}}\left(\text{\boldmath{$\omega$}},\mathcal{B}\right)$ by using the block coordinate descent method to solve a weighted MMSE (WMMSE) problem. The framework proposed in \cite{abrardo2020intelligent} allows us to obtain a local optimum of the problem:
\begin{equation} \label{P:sumRateMax1}
\max \limits_{\mathcal{B}} \int  R_{\text{tot}}\left({\text{\boldmath{$\omega$}}},\mathcal{B}\right) f_{\Omega}\left(\text{\boldmath{$\omega$}}\right) \, \text{d}\text{\boldmath{$\omega$}}.
\end{equation}
In \eqref{P:sumRateMax1}, the optimal RIS phase profile is evaluated by maximizing the conditional sum rate, where the average is evaluated with respect to the distribution $f_{\Omega}\left(\text{\boldmath{$\omega$}}\right) $ of ${\Omega}$. Specifically, the integral in \eqref{P:sumRateMax1} is computed by the Monte Carlo method, which involves generating the user locations and the NLOS channel links according to their statistical distribution. The details of the solution of \eqref{P:sumRateMax1} are omitted for clarity. Interested readers can refer to \cite{abrardo2020intelligent} for the derivation.

The proposed approach  belongs to the schemes commonly referred to as two-time scale approaches, where the optimization of the RIS is performed at time scales much longer than the coherence time of the channel. However, classical two-time scale approaches \cite{HanTanJinWenMa19, HuDaiHanWan21, ZhiPanRenWan21,Pan21} are based on the ideal assumption that the positions of the users are known exactly, i.e. $\hat{\boldsymbol{p}}_{\mathrm{U}_{i,0}} = {\boldsymbol{p}}_{\mathrm{U}_{i,0}}$ and $\boldsymbol{\Sigma}_{i,0} = \boldsymbol{0}$. In this case, the set $\mathcal{P}$ reduces to ${\boldsymbol{p}}_{\mathrm{U}_{i,0}}$, i.e., the rate in \eqref{P:sumRateMax1} is averaged over the NLOS channel realizations only. This approach can be considered as a benchmark that results in upper bound performance. To prove the effectiveness of the proposed RIS optimization approach, it is interesting to analyze the case study where the user position is affected by an estimation error, i.e., $\hat{\boldsymbol{p}}_{\mathrm{U}_{i,0}} \ne {\boldsymbol{p}}_{\mathrm{U}_{i,0}}$ but the RIS is optimized without considering this uncertainty, i.e. assuming $\boldsymbol{\Sigma}_{i,0} = \boldsymbol{0}$. Also in this case, called \emph{punctual optimization}, the set $\mathcal{P}$ reduces to a single point, i.e., the rate in \eqref{P:sumRateMax1} is averaged only over the NLOS channel realizations.

\section{Channel-coherent Optimization Phase}\label{sec_ChcoherentOpt}

In this section, we describe the operations executed in Phase II. In each channel coherence interval (indexed by $\tau$), specifically, we estimate the composed channel from the BS to each UE, and optimize the BS precoders. 
For notation convenience, we drop the index $\tau>0$. 
\subsection{Channel Estimation} \label{sec_CHEst}
Given the optimized RIS phase profile $\mathcal{B} = \{\boldsymbol{b}_1, \boldsymbol{b}_2, \cdots, \boldsymbol{b}_K \}$ obtained as described in Section \ref{sec_RISPhase}, the composed channel between the BS and the $i$-th UE is denoted by $\boldsymbol{H}_{\mathrm{B}, i} (\mathcal{B}, \boldsymbol{p}_{\mathrm{U}_i})$. To perform channel estimation, the BS transmits the pilot matrix $\boldsymbol{X}_{\mathrm{P}_2} = [\boldsymbol{x}_1, \boldsymbol{x}_2, \cdots, \boldsymbol{x}_{T_{\mathrm{P}_2}B}]$, and the corresponding received signal at the $i$-th UE is $\boldsymbol{Y}_{i, \mathrm{P}_2} = \boldsymbol{H}_{\mathrm{B}, i} (\mathcal{B}, \boldsymbol{p}_{\mathrm{U}_i}) \boldsymbol{X}_{\mathrm{P}_2} + \boldsymbol{N}_{i, \mathrm{P}_2}$, where $\boldsymbol{N}_{i, \mathrm{P}_2} = [\boldsymbol{n}_{1, i}^{\mathrm{P}_2}, \boldsymbol{n}_{2, i}^{\mathrm{P}_2}, \cdots, \boldsymbol{n}_{T_{\mathrm{P}_2}B, i}^{\mathrm{P}_2}]$ denotes the noise component, with $\boldsymbol{n}_{i, \mathrm{P}_2} = \mathrm{vec}(\boldsymbol{N}_{i, \mathrm{P}_2})\sim\mathcal{CN}(\boldsymbol{0}, \sigma_n^2 \boldsymbol{I}_{N_{\mathrm{U}_i}T_{\mathrm{P}_2}B})$. Denoting $\boldsymbol{y}_{i, \mathrm{P}_2} = \mathrm{vec}(\boldsymbol{Y}_{i, \mathrm{P}_2}) $, we have $\boldsymbol{y}_{i, \mathrm{P}_2} = \boldsymbol{X} \boldsymbol{h}_{\mathrm{B}, i} ( \mathcal{B}, \boldsymbol{p}_{\mathrm{U}_i} ) + \boldsymbol{n}_{i, \mathrm{P}_2}$, where $\boldsymbol{h}_{\mathrm{B}, i} ( \mathcal{B}, \boldsymbol{p}_{\mathrm{U}_i}) = \mathrm{vec} (\boldsymbol{H}_{\mathrm{B}, i} ( \mathcal{B}, \boldsymbol{p}_{\mathrm{U}_i}))$, and $\boldsymbol{X} = ( \boldsymbol{X}_{\mathrm{P}_2}^{\mathrm{T}} \otimes \boldsymbol{I}_{N_{\mathrm{U}_i}} )$. 

Since the end-to-end channel $\boldsymbol{H}_{\mathrm{B}, i} ( \mathcal{B}, \boldsymbol{p}_{\mathrm{U}_i})$ includes the NLOS components of BS-UE direct link and the NLOS components of RIS-UE links, which are unstructured components, channel estimators that exploit the channel sparsity are not suitable. Therefore, we consider two classical channel estimators: the maximum-likelihood (ML) channel estimator which does not require the prior statistics of the channel (such as the mean $\bar{\boldsymbol{h}}_{\mathrm{B}, i}$ and the covariance $\boldsymbol{R}_{\mathrm{B}, i}$), and the channel estimator which utilizes such prior information. In particular, thanks to the localization performed in Phase I, some partial channel state information is obtained to exploit the MMSE channel estimation method to improve the channel estimation accuracy.

\subsubsection{ML channel estimator}
ML channel estimation can be applied if $T_{\mathrm{P}_2}B \ge N_{\mathrm{B}}$. Specifically,
\begin{align}
\hat{\boldsymbol{h}}_{\mathrm{B},i}^{(\mathrm{ML})} = (\boldsymbol{X}^{\mathrm{H}} \boldsymbol{X} )^{-1} \boldsymbol{X}^{\mathrm{H}} \boldsymbol{y}_{i, \mathrm{P}_2}. \label{eqCHLS}
\end{align}
The associated channel estimation error $\Delta \boldsymbol{h}_{\mathrm{B}, i}^{(\mathrm{ML})}$ is independent of $\boldsymbol{h}_{\mathrm{B}, i} \left( \mathcal{B} \right)$, and its distribution is $\mathcal{CN} (\boldsymbol{0}, \boldsymbol{E}_{i}^{(\mathrm{ML})})$, where
\begin{align}
\boldsymbol{E}_{i}^{(\mathrm{ML})} = \sigma_n^2 ( \boldsymbol{X}^{\mathrm{H}} \boldsymbol{X} )^{-1}. \label{eqLSErr}
\end{align}

\subsubsection{LMMSE channel estimator}
According to Section \ref{sec_model}, given the optimized RIS phase profile $\mathcal{B}$ and the UE location $\boldsymbol{p}_{\mathrm{U}_i}$, the mean vector $\bar{\boldsymbol{h}}_{\mathrm{B}, i} ( \mathcal{B}, \boldsymbol{p}_{\mathrm{U}_i})$ and the covariance matrix $\boldsymbol{R}_{\mathrm{B}, i} ( \mathcal{B}, \boldsymbol{p}_{\mathrm{U}_i})$ of the channel vector $\boldsymbol{h}_{\mathrm{B}, i} ( \mathcal{B}, \boldsymbol{p}_{\mathrm{U}_i})$ are given by
\begin{align}
\bar{\boldsymbol{h}}_{\mathrm{B}, i} ( \mathcal{B}, \boldsymbol{p}_{\mathrm{U}_i}) = & \sum_{k=1}^K \bar{\boldsymbol{h}}_{k, i} \left(\boldsymbol{b}_k, \boldsymbol{p}_{\mathrm{U}_i}\right), \label{eqmeanh} \\
\boldsymbol{R}_{\mathrm{B}, i} ( \mathcal{B}, \boldsymbol{p}_{\mathrm{U}_i}) = & \sigma_{\mathrm{B}, i}^{2}  \boldsymbol{I}_{N_{\mathrm{B}}N_{\mathrm{U}_i}} + \sum_{k=1}^K \boldsymbol{R}_{k, i} \left(\boldsymbol{b}_k\right), \label{eqVarh}
\end{align}
where $\bar{\boldsymbol{h}}_{k, i} \left(\boldsymbol{b}_k, \boldsymbol{p}_{\mathrm{U}_i}\right)$ and $\boldsymbol{R}_{k, i} \left(\boldsymbol{b}_k \right)$ are defined in \eqref{eq:meanCHBRU} and \eqref{eq:varCHBRU}, respectively.

Based on the location estimates in Section \ref{sec_UELocEst}, the mean of the channel given the optimized RIS phase profile $\mathcal{B}$ is obtained by marginalizing \eqref{eqmeanh} with respect to the UE positions, as
\begin{align}
\bar{\boldsymbol{h}}_{\mathrm{B}, i} ( \mathcal{B}) = & \sum_{k=1}^K \int \bar{\boldsymbol{h}}_{k, i} \left(\boldsymbol{b}_k, \boldsymbol{p}_{\mathrm{U}_i}\right) f \left( \boldsymbol{p}_{\mathrm{U}_i} \right) \text{d} \boldsymbol{p}_{\mathrm{U}_i}. \label{eqmh}
\end{align}
Here, $f( \boldsymbol{p}_{\mathrm{U}_i} )$ is the probability density function of the $i$-th UE location, which is assumed to follow a Gaussian distribution with mean $\boldsymbol{\hat p}_{\mathrm{U}_i, 0}$ and covariance matrix ${\boldsymbol{\Sigma}}_{i,0}$. Similarly, the covariance matrix of the channel vector is
\begin{align}
\boldsymbol{R}_{\mathrm{B}, i} ( \mathcal{B} )
= & \int \mathrm{E} \left\{ \boldsymbol{h}_{\mathrm{B}, i} ( \mathcal{B}, \boldsymbol{p}_{\mathrm{U}_i}) \boldsymbol{h}_{\mathrm{B}, i}^{\mathrm{H}} ( \mathcal{B}, \boldsymbol{p}_{\mathrm{U}_i}) \right\} f\left( \boldsymbol{p}_{\mathrm{U}_i} \right) \text{d} \boldsymbol{p}_{\mathrm{U}_i} \notag \\
& \hspace{34mm}-\bar{\boldsymbol{h}}_{\mathrm{B}, i} ( \mathcal{B})\bar{\boldsymbol{h}}_{\mathrm{B}, i}^{\mathrm{H}} ( \mathcal{B}). \label{eqvarh}
\end{align}

The mean and covariance in \eqref{eqmh} and \eqref{eqvarh} can be computed numerically given the statistical characterization of the UE positions. Also, when $\sqrt{\mathrm{Tr} ( {\boldsymbol{\Sigma}}_{i,0} ) }$ is small, we can assume, based on Appendix \ref{GaussianAssumption}, that $\boldsymbol{h}_{\mathrm{B}, i} ( \mathcal{B}, \boldsymbol{p}_{\mathrm{U}_i})$ has a Gaussian distribution. We have the following proposition.

% \textcolor{blue}{Good notes: We need to justify in which condition that the assumption is reasonable. Also, we can use the numerical results to show that the assumption is OK. (This is important as the rate computation requires the Gaussian justification.)}

\begin{proposition}\label{prop1}
Assume the channel vector $\boldsymbol{h}_{\mathrm{B}, i} ( \mathcal{B}, \boldsymbol{p}_{\mathrm{U}_i} ) \sim \mathcal{CN} (\bar{\boldsymbol{h}}_{\mathrm{B}, i} ( \mathcal{B}), \boldsymbol{R}_{\mathrm{B}, i} ( \mathcal{B} ))$, the LMMSE estimator for $\boldsymbol{h}_{\mathrm{B}, i}$ is given by
\begin{align}
\hat{\boldsymbol{h}}_{\mathrm{B}, i}^{(\mathrm{MMSE})} = & \boldsymbol{\Lambda}_i \left( \boldsymbol{y}_{i, \mathrm{P}_2} - \boldsymbol{X}\bar{\boldsymbol{h}}_{\mathrm{B}, i} \right) + \bar{\boldsymbol{h}}_{\mathrm{B}, i}, \label{eqLMMSE}
\end{align}
where $\boldsymbol{\Lambda}_i =  \boldsymbol{R}_{\mathrm{B}, i} \boldsymbol{X}^{\mathrm{H}} ( \boldsymbol{X} \boldsymbol{R}_{\mathrm{B}, i} \boldsymbol{X}^{\mathrm{H}} + \sigma_n^2 \boldsymbol{I}_{N_{\mathrm{U}_i}T_{\mathrm{P}_2}} )^{-1}$ %is given by
%\begin{align}
%\boldsymbol{\Lambda}_i = & \boldsymbol{R}_{\mathrm{B}, i} \boldsymbol{S}^{\mathrm{H}} \left( \boldsymbol{S} \boldsymbol{R}_{\mathrm{B}, i} \boldsymbol{S}^{\mathrm{H}} + \sigma_n^2 \boldsymbol{I}_{N_{\mathrm{U}_i}T_{\mathrm{P}_2}} \right)^{-1}, \label{eqLMMSEMat}
%\end{align}
and the mean-square error $\Delta \boldsymbol{h}_{\mathrm{B}, i}^{(\mathrm{MMSE})}$ satisfies $\mathrm{E}\{\Delta \boldsymbol{h}_{\mathrm{B}, i}^{(\mathrm{MMSE})} (\hat{\boldsymbol{h}}_{\mathrm{B}, i}^{(\mathrm{MMSE})})^{\mathrm{H}} \} = \boldsymbol{0}$ %$\hat{\boldsymbol{h}}_{\mathrm{B}, i}^{(\mathrm{MMSE})}$,
with $\mathcal{CN}(\boldsymbol{0}, \boldsymbol{E}_{i}^{(\mathrm{MMSE})})$ where
\begin{align}
\boldsymbol{E}_{i}^{(\mathrm{MMSE})} = \boldsymbol{R}_{\mathrm{B}, i} - \boldsymbol{\Lambda}_i \boldsymbol{X} \boldsymbol{R}_{\mathrm{B}, i}. \label{eqMMSEErr}
\end{align}
\end{proposition}
% \begin{proof}
% The proof can be found in Appendix \ref{AppendixA}
% \end{proof}
The proof of Proposition \ref{prop1} is available in \cite{Kay:93, Pan21}. The relation between the ML and LMMSE estimators is elaborated in the following proposition.
\begin{proposition} \label{prop2}
Let $\Delta\boldsymbol{E}_i = \boldsymbol{E}_{i}^{(\mathrm{ML})} - \boldsymbol{E}_{i}^{(\mathrm{MMSE})}$ where $\boldsymbol{E}_{i}^{(\mathrm{MMSE})}$ and $\boldsymbol{E}_{i}^{(\mathrm{ML})}$ are the covariance matrices of the channel estimation error given by \eqref{eqMMSEErr} and \eqref{eqLSErr}, respectively. We then have $\Delta\boldsymbol{E}_i  \succ \boldsymbol{0}$, i.e., $\Delta\boldsymbol{E}_i$ is positive definite. 
\end{proposition}
\begin{proof}
The proof immediately follows from \eqref{eqLSErr} and \eqref{eqMMSEErr}.

\end{proof}
Proposition \ref{prop2} indicates that, thanks to the knowledge of the statistics of the channel, we can always achieve better channel estimates with the LMMSE estimator compared with the ML estimator. Therefore, we consider the LMMSE channel estimator in the following.

\subsection{Conditional Achievable Rate}
During the data transmission phase, the signal received at the $i$-th UE in the presence of $N_{\mathrm{U}}$ concurrent transmitted streams is given by
\begin{align}
\boldsymbol{y}_i^{(\mathrm{D})} = & \boldsymbol{H}_{\mathrm{B}, i} \boldsymbol{x}_i + \boldsymbol{H}_{\mathrm{B}, i} \sum_{j=1, j\ne i}^{N_{\mathrm{U}}} \boldsymbol{x}_j  + \boldsymbol{n}_i^{(\mathrm{D})}, \label{eqRecData}
\end{align} 
where $\boldsymbol{x}_j = \boldsymbol{V}_j\boldsymbol{s}_j^{(\mathrm{D})}$ is the transmitted precoded vector, $\boldsymbol{V}_j$ is the precoding matrix, and $\boldsymbol{s}_j^{(\mathrm{D})}$ is the vector of transmitted symbols to the $j$-th UE, with normalized power, i.e., $\mathrm{E} \{ \boldsymbol{s}_j^{(\mathrm{D})} (\boldsymbol{s}_j^{(\mathrm{D})})^{\mathrm{H}} \} = \boldsymbol{I}_{N_{\mathrm{U}}}$. The noise component $\boldsymbol{n}_i^{(\mathrm{D})}$ is modeled as $\mathcal{CN}(\boldsymbol{0}, \sigma_n^2 \boldsymbol{I}_{N_{\mathrm{U}}})$. 
\subsubsection{Transmitter with complete channel knowledge}
When the complete CSI is known at the transmitter, the achievable rate of the $i$-th UE can be obtained as \cite{abrardo2020intelligent}
\begin{align}
R_i \left(\mathcal{V}, \mathcal{B} \right) = & \log \det \left(\boldsymbol{I}_{N_{\mathrm{U}}} + \boldsymbol{V}_i^{\mathrm{H}} \boldsymbol{H}_{\mathrm{B}, i}^{\mathrm{H}} \left(\mathcal{B} \right) \bar{\boldsymbol{J}}_i^{-1} \boldsymbol{H}_{\mathrm{B}, i}\left(\mathcal{B} \right) \boldsymbol{V}_i \right), \label{eqARPerfectCH}
\end{align}
with $\bar{\boldsymbol{J}}_i =  \sum_{j=1, j\ne i}^{N_{\mathrm{U}}} \boldsymbol{H}_{\mathrm{B}, j} \boldsymbol{V}_j \boldsymbol{V}_j^{\mathrm{H}} \boldsymbol{H}_{\mathrm{B}, j}^{\mathrm{H}} + \sigma_n^2 \boldsymbol{I}_{N_{\mathrm{U}}}$, where $\mathcal{V}=\left\{\mathbf{V}_{1},\mathbf{V}_{2},\dots,\mathbf{V}_{N_{u}}\right\}$ denotes the set of precoding matrices of all the $N_{\mathrm{U}}$ UEs.

\subsubsection{Transmitter with LMMSE channel estimates}
When the LMMSE channel estimates are available at the transmitter, we rewrite \eqref{eqRecData} as
\begin{align}
\boldsymbol{y}_i^{(D)} = & \big( \hat{\boldsymbol{H}}_{\mathrm{B}, i}^{(\mathrm{MMSE})} + \Delta \boldsymbol{H}_{\mathrm{B}, i}^{(\mathrm{MMSE})} \big) \sum_{j=1}^{N_{\mathrm{U}}} \boldsymbol{x}_j + \boldsymbol{n}_i^{(\mathrm{D})} \notag\\
= & \hat{\boldsymbol{H}}_{\mathrm{B}, i}^{(\mathrm{MMSE})} \boldsymbol{x}_i + \tilde{\boldsymbol{n}}_i, \label{eqRecMMSEchest}
\end{align}
where $\hat{\boldsymbol{h}}_{\mathrm{B}, i}^{(\mathrm{MMSE})} = \mathrm{vec} (\hat{\boldsymbol{H}}_{\mathrm{B}, i}^{(\mathrm{MMSE})} )$ is given in \eqref{eqLMMSE}, and $\Delta \boldsymbol{h}_{\mathrm{B}, i}^{(\mathrm{MMSE})} = \mathrm{vec} (\Delta {\boldsymbol{H}}_{\mathrm{B}, i}^{(\mathrm{MMSE})} )$ is the LMMSE channel estimation error vector. $\tilde{\boldsymbol{n}}_i =  \Delta \boldsymbol{H}_{\mathrm{B}, i}^{(\mathrm{MMSE})} \boldsymbol{x}_i + \boldsymbol{H}_{\mathrm{B}, i} \sum_{j=1, j\ne i}^{N_{\mathrm{U}}} \boldsymbol{x}_j  + \boldsymbol{n}_i^{(\mathrm{D})}$ is the equivalent noise-plus-interference component. The corresponding conditional achievable rate is given next.
\begin{proposition} \label{prop3}
With the LMMSE channel estimates $\hat{\boldsymbol{h}}_{\mathrm{B}, i}^{(\mathrm{MMSE})}$ given in \eqref{eqLMMSE}, the conditional achievable rate of the $i$-th UE is
\begin{align}
R_i \left(\mathcal{V}, \mathcal{B} \right) = & \log \det (\boldsymbol{I}_{N_{\mathrm{U}}} + \boldsymbol{V}_i^{\mathrm{H}} \hat{\boldsymbol{H}}_{\mathrm{B}, i}^{\mathrm{H}} (\mathcal{B}) \tilde{\boldsymbol{J}}_i^{-1} \hat{\boldsymbol{H}}_{\mathrm{B}, i}(\mathcal{B}) \boldsymbol{V}_i ), \label{eqARmmseCH}
\end{align}
with $\tilde{\boldsymbol{J}}_i$ given by
\begin{align}
\tilde{\boldsymbol{J}}_i = & \sigma_n^2 \boldsymbol{I}_{N_{\mathrm{U}}} + E_{\mathrm{s}}\sum_{m=1}^{N_{\mathrm{B}}} \sum_{n=1}^{N_{\mathrm{B}}} \Pi_i^{(m,n)} \boldsymbol{E}_i^{(m, n)} \notag \\ & +E_{\mathrm{s}} \sum_{m=1}^{N_{\mathrm{B}}} \sum_{n=1}^{N_{\mathrm{B}}} \big( \sum_{j=1, j\ne i}^{N_{\mathrm{U}}} \Pi_j^{(m,n)} \big) \boldsymbol{R}_{\mathrm{B}, i}^{(m, n)} , \label{eqtildeJi}
\end{align}
where $\boldsymbol{E}_i^{(m,n)} = \mathrm{E} \{\Delta \boldsymbol{h}_{\mathrm{B}, i}^{(m)} ( \Delta \boldsymbol{h}_{\mathrm{B}, i}^{(n)} )^{\mathrm{H}} \}$, $\boldsymbol{R}_{\mathrm{B}, i}^{(m, n)} = \mathrm{E} \{ \boldsymbol{h}_{\mathrm{B}, i}^{(m)} (\mathcal{B}, \boldsymbol{p}_{\mathrm{U}_i}) (\boldsymbol{h}_{\mathrm{B}, i}^{(m)} (\mathcal{B}, \boldsymbol{p}_{\mathrm{U}_i}))^{\mathrm{H}} \}$, and $\Pi_i^{(m,n)}$ is the ($m,n$)-th element in $\boldsymbol{\Pi}_i= \boldsymbol{V}_i \boldsymbol{V}_i^{\mathrm{H}}$.
\end{proposition}
\begin{proof}
See Appendix \ref{ProofProp3}.
\end{proof}
Equations \eqref{eqARPerfectCH}\,and\,\eqref{eqARmmseCH} are used to optimize the RIS phase profile in Phase I and the BS precoding in Phase III, respectively. As for the optimization of the RIS phase profile, from the rate in \eqref{eqARPerfectCH} or \eqref{eqARmmseCH}, the 
conditional sum-rate in \eqref{P:sumRateMax1} can be formulated as $R_{\text{tot}}\left({\text{\boldmath{$\omega$}}},\mathcal{B}\right) = \sum_{i} R_i \left(\mathcal{V}, \mathcal{B} \right)$. Then, the integral in \eqref{P:sumRateMax1} is computed numerically, from a set of realizations for $\Omega$ and computing \eqref{eqARPerfectCH} or \eqref{eqARmmseCH} for each realization. The optimization of the BS precoding is discussed next.

\subsection{Optimal Precoder Design} \label{sec_PreOpt}
In this section, we compute the optimal precoding matrices by taking into account that the LMMSE channel estimates are available at the transmitter. To the best of our knowledge, this is not available in the open technical literature. On the other hand, the computation of the optimal precoding matrices with complete channel knowledge is discussed in \cite{abrardo2020intelligent}. The optimal precoding matrices can be derived according to the following sum-rate maximization problem
\begin{align}
\label{P:sumRateMax_pre}
&\max \limits_{\mathcal{V}} \sum\limits_{i = 1}^{N_{\mathrm{U}}} R_i \left(\mathcal{V}, \mathcal{B} \right) \\
\text{s.t. } & \text{Tr} \left(\mathbf{V}_i\mathbf{V}^{\mathrm{H}}_i \right) \le P_{i}, ~~ i=1,2,\dots,N_{\mathrm{U}}, \notag
\end{align}
where $P_i$ is the power budget of the $i$-th UE and $R_i \left(\mathcal{V}, \mathcal{B} \right)$ is defined in \eqref{eqARmmseCH}. To solve problem \eqref{P:sumRateMax_pre}, we utilize the iterative WMMSE algorithm. Let us introduce
the MSE matrix $\boldsymbol{E}_{i}(\mathcal{V},\boldsymbol{G}_{i}) = \text{E}_{\boldsymbol{s},\boldsymbol{n},\Delta \boldsymbol{H}_i}\big\{ ( {\boldsymbol{s}}_i - \boldsymbol{\tilde s}_i )( \boldsymbol{s}_i - \boldsymbol{\tilde s}_i )^{\mathrm{H}} \big\} $, where $\boldsymbol{\tilde s}_i =  \boldsymbol{G}^{\mathrm{H}}_{i} \boldsymbol{y}_{i}$ and $\boldsymbol{G}_{i}$ is the linear decoding matrix of the $i$th UE. The superscript `$(\mathrm{D})$' is omitted for simplicity. Assuming that the information symbols are zero-mean and {i.i.d.} RVs, i.e., $\mathrm{E}\left \{ \boldsymbol{s}_i \boldsymbol{s}_i^{\mathrm{H}}\right \} = \boldsymbol{I}_{N_{\mathrm{U}}}$ and $\mathrm{E}\left \{ \boldsymbol{s}_i\boldsymbol{s}_j^{\mathrm{H}}\right \} =  \mathbf{0}_{L}$ for $j \ne i$, we get
 \begin{equation}
\label{eq:mMSEMat0_15}
\begin{aligned}
\boldsymbol{E}_{i}(\mathcal{V},\boldsymbol{G}_{i}) & = (\boldsymbol{I} - \boldsymbol{G}_{i}^{\mathrm{H}} \hat{\boldsymbol{H}}_i \boldsymbol{V}_i )(\boldsymbol{I} - \boldsymbol{G}_{i}^{\mathrm{H}} \hat{\boldsymbol{H}}_i \boldsymbol{V}_i )^{\mathrm{H}} \notag \\ & \hspace{-8mm} + \sum_{j=1, j\ne i}^{N_\mathrm{U}} \boldsymbol{G}_{i}^{\mathrm{H}} \hat{\boldsymbol{H}}_i \boldsymbol{V}_j\boldsymbol{V}_j^{\mathrm{H}} \hat{\boldsymbol{H}}_{i}^{\mathrm{H}} \boldsymbol{G}_{i} \\
&\hspace{-8mm} + \sum_{j=1}^{N_\mathrm{U}} \text{E}_{\Delta \boldsymbol{H}_i}\{\boldsymbol{G}_{i}^{\mathrm{H}} \Delta{\boldsymbol{H}}_i \boldsymbol{V}_j\boldsymbol{V}_j^{\mathrm{H}} \Delta{\boldsymbol{H}}_i^{\mathrm{H}} \boldsymbol{G}_{i}\} + \sigma_n^2 \boldsymbol{G}_{i}^{\mathrm{H}} \boldsymbol{G}_{i} \\
& \hspace{-12mm} =  (\boldsymbol{I} - \boldsymbol{G}_{i}^{\mathrm{H}} \hat{\boldsymbol{H}}_i \boldsymbol{V}_i )(\boldsymbol{I} - \boldsymbol{G}_{i}^{\mathrm{H}} \hat{\boldsymbol{H}}_i \boldsymbol{V}_i )^{\mathrm{H}} \\ 
& \hspace{-8mm}+ \sum_{j=1, j\ne i}^{N_\mathrm{U}} \boldsymbol{G}_{i}^{\mathrm{H}} \hat{\boldsymbol{H}}_i \boldsymbol{V}_j\boldsymbol{V}_j^{\mathrm{H}} \hat{\boldsymbol{H}}_{i}^{\mathrm{H}} \boldsymbol{G}_{i} +  \boldsymbol{G}_{i}^{\mathrm{H}} (\boldsymbol{C}_i + \sigma_n^2 \boldsymbol{I})\boldsymbol{G}_{i}  
\end{aligned}
\end{equation}
where $\boldsymbol{C}_{i} = \sum_{j=1}^{N_\mathrm{U}} \mathrm{E}_{\Delta \boldsymbol{H}_i} \{ \Delta{\boldsymbol{H}}_i \boldsymbol{V}_j \boldsymbol{V}_j^{\mathrm{H}} \Delta{\boldsymbol{H}}_i^{\mathrm{H}} \}$ and $\mathrm{E}_{\Delta \boldsymbol{H}_i} \{ \Delta{\boldsymbol{H}}_i \boldsymbol{V}_j \boldsymbol{V}_j^{\mathrm{H}} \Delta{\boldsymbol{H}}_i^{\mathrm{H}} \}$ is given in \eqref{eqfirComp}. 

Problem \eqref{P:sumRateMax_pre} can be reformulated as the following MSE minimization problem
\begin{align}
&\min \limits_{\mathcal{V},\mathcal{W},\mathcal{G}} \sum\limits_{i = 1}^{N_{\mathrm{U}}} \big\{\text{Tr}\left[\boldsymbol{W}_{i}\boldsymbol{E}_{i}\left(\mathcal{V},\boldsymbol{G}_{i}\right)\right]-\log\det\left(\boldsymbol{W}_{i}\right)\big\} \label{wMMSE1_1}\\
&\quad \text{s.t.}
\quad \text{Tr} \left(\boldsymbol{V}_i\boldsymbol{V}^{\mathrm{H}}_i\right)\le P_{i}, ~~ i=1,2,\dots,N_{\mathrm{U}}, \notag
\end{align}
where $\boldsymbol{W}_{i}\succeq 0$ is the matrix of weights for the MSE of the $i$-th UE, and $\mathcal{W}=\{{\boldsymbol{W}}_{1},\boldsymbol{W}_{2},\dots,\boldsymbol{W}_{N_{\mathrm{U}}} \}$
and $\mathcal{G}=\left\{\boldsymbol{G}_{1},\boldsymbol{G}_{2},\dots,\boldsymbol{G}_{N_{\mathrm{U}}}\right\}$ are the sets of all weight and receive filter matrices, respectively.

The equivalence between problem \eqref{P:sumRateMax_pre} and \eqref{wMMSE1_1} follows by recalling the relation between the MMSE covariance $\hat{\boldsymbol{E}}_{i}(\mathcal{V}) = \min\limits_{\boldsymbol{G}_{i}}\boldsymbol{E}_{i}(\mathcal{V},\boldsymbol{G}_{i})$, the achievable rate %$I_i$, i.e., 
$ \log \det(\hat{\boldsymbol{E}}_{i} ^{-1}(\mathcal{V}))$ and by the fact that the optimal solution of \eqref{wMMSE1_1} is $\boldsymbol{W}_{i} = \hat{\boldsymbol{E}}^{-1}_{i} (\mathcal{V} )$. Further details can be found in \cite{Shi2011}.

Unfortunately, problem \eqref{wMMSE1_1} is non-convex. If all the optimization variables are fixed except one, however, it is a convex optimization problem in the remaining variables. Accordingly, the weighted MSE minimization problem in \eqref{wMMSE1_1} can be solved by using an iterative block coordinate descent (BCD) algorithm \cite{Bertsekas1999}. To elaborate, let us denote by  $\boldsymbol{G}_{i}^{(q+1)}$, $\boldsymbol{W}_{i}^{(q+1)}$ and $\boldsymbol{V}_{i}^{(q+1)}$ the optimization variables after the $(q+1)$-th iteration. Then, we have the following.
\begin{itemize}
\item Receive filter: The receive filter matrix can be computed as $\boldsymbol{G}_{i}^{(q+1)} = \arg\min \limits_{\boldsymbol{G}_i} \boldsymbol{E}_{i}(\mathcal{V}^{(q)},\boldsymbol{G}_{i})$, which yields $\boldsymbol{G}_{i}^{(q+1)}=  (\boldsymbol{J}_{i}^{(q)})^{-1}\hat{\boldsymbol{H}}_i\boldsymbol{V}_{i}^{(q)}$, in which $\boldsymbol{J}_{i}^{(q)}=\sum_{j = 1}^{N_{\mathrm{U}}}\hat{\boldsymbol{H}}_i\boldsymbol{V}_{j}^{(q)}(\boldsymbol{V}_{j}^{(q)})^{\mathrm{H}}\hat{\boldsymbol{H}}_{i}^{\mathrm{H}}  + \boldsymbol{C}_i + \sigma_n^2 \boldsymbol{I}$. The corresponding MMSE is $\boldsymbol{E}_{i}^{(q+1)} = \boldsymbol{E}_{i}(\mathcal{V},\boldsymbol{G}_{i}^{(q+1)})  = \boldsymbol{I} - \boldsymbol{V}^{\mathrm{H}}_i \hat{\boldsymbol{H}}^{\mathrm{H}}_i(\boldsymbol{J}_{i}^{(q)})^{-1} \hat{\boldsymbol{H}}_i  \boldsymbol{V}_i.$

\item Weights: The weights can be computed as $\boldsymbol{W}_{i}^{(q+1)}=(\boldsymbol{E}_{i}^{(q+1)})^{-1}$. 
\item Precoding filters: $\boldsymbol{V}_{i}^{q+1}$ can be computed solving the following problem
\begin{align*}
&\boldsymbol{V}_{i}^{q+1} = \arg\min \limits_{\mathcal{V}} \text{Tr}\Big(\boldsymbol{W}_{i}^{(q+1)} \left(\boldsymbol{I} - \boldsymbol{\Gamma}_{i,i}^{(q+1)} \right) (\boldsymbol{I} - \boldsymbol{\Gamma}_{i,i}^{(q+1)} )^{\mathrm{H}} \\
& + \sum_{j=1, j\ne i}^{N_\mathrm{U}} \boldsymbol{W}_{j}^{(q+1)}  \boldsymbol{\Gamma}_{j,i}^{(q+1)}\left( \boldsymbol{\Gamma}_{j,i}^{(q+1)} \right)^{\mathrm{H}} + \sum_{j=1}^{N_\mathrm{U}} \boldsymbol{W}_{j}^{(q+1)} \\ & \cdot \text{E}_{\Delta \boldsymbol{H}_j}\left\{\left(\boldsymbol{G}_{j}^{(q+1)}\right)^{\mathrm{H}} \Delta{\boldsymbol{H}}_j \boldsymbol{V}_i\boldsymbol{V}_i^{\mathrm{H}} \Delta{\boldsymbol{H}}_j^{\mathrm{H}} \boldsymbol{G}_{j}^{(q+1)}\right\}\Big) \numberthis \label{Prec_1}\\
&\quad \text{s.t.}
\quad \text{Tr} \left(\boldsymbol{V}_i\boldsymbol{V}^{\mathrm{H}}_i\right)\le P_{i}, ~~ i=1,2,\dots,N_{\mathrm{U}},  
\end{align*}
in which $\boldsymbol{\Gamma}_{i,j}^{(q)} = \big(\boldsymbol{G}_i^{(q)}\big)^{\mathrm{H}} \hat{\boldsymbol{H}}_i\boldsymbol{V}_j$. The solution to problem \eqref{Prec_1} is $\boldsymbol{V}_{i}^{(q+1)} = (\boldsymbol{K}^{(q+1)}+ \mu_i \boldsymbol{I}_{M} )^{-1}\hat{\boldsymbol{H}}_i^H\boldsymbol{G}_{i}^{(q+1)}\boldsymbol{W}_{i}^{(q+1)}$ where 
\begin{align*}
& \boldsymbol{K}^{(q+1)}= \sum\limits_{j = 1}^{N_{\mathrm{U}}} \hat{\boldsymbol{H}}_j^{\mathrm{H}} \boldsymbol{G}_{j}^{(q+1)} \boldsymbol{W}_{j}^{(q+1)} (\boldsymbol{G}_{j}^{(q+1)})^{\mathrm{H}} \hat{\boldsymbol{H}}_j \\
+ & \text{E}_{\Delta \boldsymbol{H}_j}\left\{\Delta{\boldsymbol{H}}_j^{\mathrm{H}}  \boldsymbol{G}_{j}^{(q+1)} \boldsymbol{W}_{j}^{(q+1)} \left( \boldsymbol{G}_{j}^{(q+1)} \right)^{\mathrm{H}} \Delta{\boldsymbol{H}}_j\right\}, \numberthis \label{Prec_2}
\end{align*}
and the Lagrange multiplier $ \mu_i$ in the optimization problem is chosen so that the power constraint in \eqref{wMMSE1_1} for the $i$-th UE is fulfilled.
\end{itemize}

\subsection{Complexity Analysis}
The proposed integrated localization and communication method encompasses four tasks: (1) localization, (2) RIS optimization, (3) channel estimation, and (4) BS precoder design. Channel estimation and BS precoder design result in the highest computational complexity, which originates from computing the inversion of matrices, whose sizes depend on the number of transmit and receive antennas. To perform localization, the proposed estimator applies two one-dimensional searches per RIS to find the AOA and the AOD, which are used to provide a coarse estimate of the location. Also, the refined estimate  of the position is obtained via an optimization in a three-dimensional space around the coarse location estimate. The computational complexity is determined by the optimization of the RIS phase shift profile, since multiple RISs equipped with a large number of elements are considered. Also, an iterative optimization method is applied in (11), whose computation complexity is discussed in \cite{abrardo2020intelligent}. However, the RISs need to be optimized on a longer timescale, i.e., every location coherence interval.

\section{Numerical Results and Discussion}\label{sec_results}
%----------------------------------------- TABLE I: Simulation Parameters -----------------------------------------------
\begin{table}%[!htb]
\caption{\textsc{Parameters in Simulations}}
\centering
    \begin{tabular}{c|c}
    \hline \hline
    \textbf{Parameters} & \textbf{Values} \\
    \hline
    Carrier frequency $f_\mathrm{c}$ & 28 GHz \\
    \hline
    Bandwidth $B$ & 120 KHz/36 MHz \\
    \hline
    Antenna array at BS & $8\times 2$ \\
    \hline
    Antenna array at UE & $2\times 2$ \\
    \hline 
    {Number of RISs, $K$} & 2 \\
    \hline
    {Number of unit cells per RIS, $P$} & $80 \times 40$ \\
    \hline 
    Antenna/cell spacing $\Delta d$ & Half wavelength \\
    \hline
    BS location $\boldsymbol{p}_{\mathrm{B}}$ & $[60, 15, 2]^{\mathrm{T}}$ \\
    \hline 
    UE location $\boldsymbol{p}_{\mathrm{U}_i}$ & $[10, 5, 0]^{\mathrm{T}}$ and $[25, 10, 0]^{\mathrm{T}}$ \\
    \hline
    {Surfaces location} $\boldsymbol{p}_{\mathrm{R}_k}$ & $[0, 15, 3]^{\mathrm{T}}$ and $[15, 20, 3]^{\mathrm{T}}$ \\
    \hline
    Transmit antenna gain $G_{\mathrm{T}}$ & 2.5 \\
    \hline
    Receive antenna gain $G_{\mathrm{R}}$ & 2.5 \\
    \hline 
    The  boresight  gain at RIS $G_{\mathrm{c}}$ & $\pi$ \\
    \hline
    Exponential parameter $q$ in \eqref{eqnprp} & 0.57 \\
    \hline
    Channel coherent interval $T_{\mathrm{c}}$ & 1 ms \\
    \hline
    Total transmission power & 22 dBm \\
    \hline
    Noise factor & 5 dB \\
    \hline
    Noise power spectrum density & -169 dBm/Hz \\
    \hline \hline
    \end{tabular}
    \label{Tab1}
\end{table} 
\subsection{Scenario}
We consider an integrated localization and communication scenario where a single BS serves two UEs. The BS and UEs are equipped with URAs, with size $8\times 2$ and $2\times 2$, respectively. In the scenario under study, we place 2 RIS on the walls surrounding the UEs. Each RIS has a size of 0.42 m $\times$ 0.21 m and forms a 80 $\times$ 40 URA. The location of the BS is $\boldsymbol{p}_{\mathrm{B}} = [60, 15, 2]^{\mathrm{T}}$. The centers of two RIS URAs are located at $\boldsymbol{p}_{\mathrm{R}_1} = [0, 15, 3]^{\mathrm{T}}$ and $\boldsymbol{p}_{\mathrm{R}_2} = [15, 20, 3]^{\mathrm{T}}$. In Phase I, the UEs are located in $\boldsymbol{p}_{\mathrm{U}_{1,0}} = [10, 5, 0]^{\mathrm{T}}$ and $\boldsymbol{p}_{\mathrm{U}_{2,0}} = [25, 10, 0]^{\mathrm{T}}$ at time 0. The location uncertainty covariance matrix in \eqref{PosPri} is set to ${\boldsymbol{\Sigma}}_{i,0}^{\mathrm{pri}} = \mathrm{Diag}([2, 2, 0])$, i.e., the location uncertainty range is $\sqrt{\mathrm{Tr} ( {\boldsymbol{\Sigma}}_{i,0}^{\mathrm{pri}} ) } = 2$ meters.\footnote{This is the reported accuracy obtained with the fingerprinting-based indoor localization algorithm in \cite{IndoorLoc16}.} The other simulation parameters are given in Table \ref{Tab1}.\footnote{The NLOS components between the RIS and UE are follows a Rician distribution with independent and identically distributed fading, similar to \cite{ZhiPanRenWan21,Pan21,DaiZhuPanRenWan22}. However, the proposed framework can be applied to other fading models that account for the spatial correlation among the RIS elements if their inter-distances are smaller than half-wavelength \cite{BjoSan21}.}

\begin{figure}
    \centering
    \begin{subfigure}[t]{0.5\textwidth}      
    \centering 
    \includegraphics[width=\linewidth]{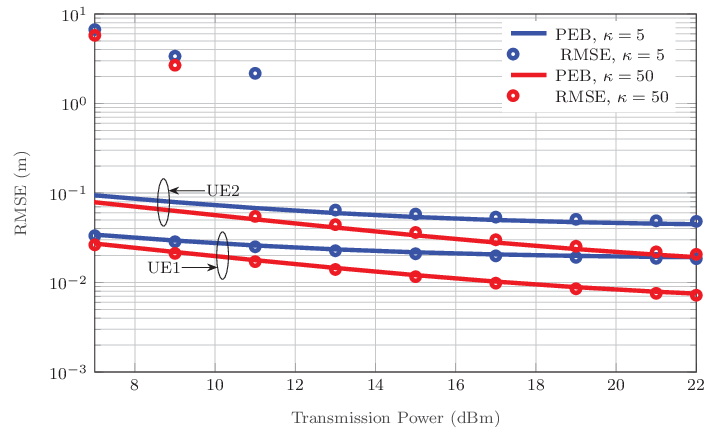}
    \caption{}
    \label{fig:PEB_Kappa}
    \end{subfigure}%
    \hfill
    \begin{subfigure}[t]{0.5\textwidth}
    \centering 
    \includegraphics[width=\linewidth]{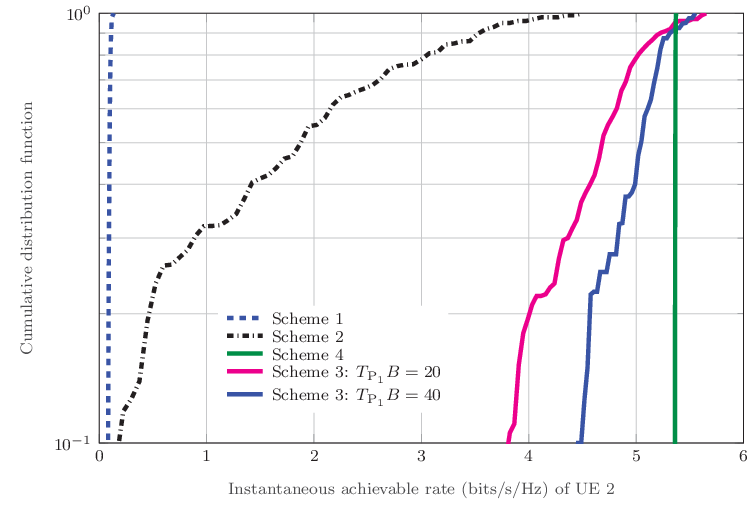}
    \caption{}
    \label{fig:cdf_AR_UE2}
    \end{subfigure}
    \caption{(a) Position error bounds and localization performance as a function of the Rician factor $\kappa$ (40 pilots are used); (b) The cumulative distribution function of the achievable rate for different RIS phase profiles.} \label{R1}
\end{figure}

For the mobility model, we assume that the UE moves only in the $x-y$ plane, with $ \boldsymbol{\Sigma}_{\text{pos}}=\mathrm{Diag}(\sigma_x^2, \sigma_y^2, 0)$. We consider that the maximum velocity of the UE is 1 m/s, which corresponds to a maximum Doppler shift $f_{\mathrm{D}}\approx 93$ Hz at the considered carrier frequency $f_c = 28$ GHz. As a result, we can set the channel coherent interval to $T_{\mathrm{C}}=1$ ms (as $f_{\mathrm{D}} T_{\mathrm{C}} \approx 0.093$), which accommodates 120 symbols for the considered bandwidth $B=120$ kHz.\footnote{In the present study, we assume that the localization phase is executed by using a single subcarrier. During the communication phase, on the other hand, we consider that multiple subcarriers are used for data transmission. Thus, the power is distributed uniformly over a bandwidth of 36 MHz (300 subcarriers). This means that the power per subcarrier is much lower during the communication phase (and the corresponding computation of the rate) as compared with the localization phase.} We assume $T_{\mathrm{L}} = 1$ s, i.e., 1000 channel coherent intervals are contained in one location coherent interval. 

\subsection{Location-coherent Optimization Phase}

\subsubsection{Location estimation performance} 
In Fig. \ref{fig:PEB_Kappa}, the position error bounds derived using \eqref{PosPhase1} are evaluated. Specifically, we see that the derived bounds are achievable with the used localization algorithm if sufficient transmission power is allocated. It is also apparent from Fig. \ref{fig:PEB_Kappa} that, if the Rician factor is sufficiently large, the position error bounds and the actual localization performance are improved. This is because the RIS-UE channel is more position dependent when the Rician factor is large. Accordingly, the advantages of the proposed localization scheme are more apparent in RIS optimization, channel estimation, and BS precoder optimization. For simplicity, we consider large Rician factors (e.g., 50) in most of the simulations. The good agreement between the position error bounds and the numerical estimates obtained by using a practical localization algorithm justifies the use of the FIM for RIS optimization, channel estimation, and BS precoder optimization.

\begin{algorithm}[tb]
\caption{RIS Optimization and Evaluation Procedure}\label{Alg1}
\begin{algorithmic} \footnotesize
\Require \parbox[t]{\dimexpr\linewidth- \algorithmicindent * 1}{${\boldsymbol{p}}_{\mathrm{U}_{i,0}}$, $\boldsymbol{\Sigma}_{i,0}$, $\boldsymbol{\Sigma}_{\mathrm{pos}}$, $\sigma_{k,i}^2$ and $\sigma_{\mathrm{B}, i}^2$ \strut}
\Repeat
\State Compute $\hat{\boldsymbol{p}}_{\mathrm{U}_{i,0}} = \boldsymbol{p}_{\mathrm{U}_{i,0}} + \boldsymbol{w}_i$ and  $\boldsymbol{p}_{\mathrm{U}_{i,\tau}} = 
\boldsymbol{p}_{\mathrm{U}_{i,0}} + \boldsymbol{v}_{i,\tau}$ (see Section \ref{sec_UELocEst});
\State \textbf{RIS optimization:}
\State 1. Obtain the optimal RIS phase profile $\mathcal{B}^{\text{opt}}$ according to Section \ref{sec_RISPhase}.
\State \textbf{Achievable rate evaluation at time $\tau$}:
\State 2. Generate the LOS component of the RIS-UE and BS-UE links based on the actual UE positions ${\boldsymbol{p}}_{\mathrm{U}_{i,\tau}}$;
\State 3. Generate the instances of the NLOS components of the RIS-UE and BS-UE links based on $\sigma_{k, i}^2$ and $\sigma_{\mathrm{B}, i}^2$, respectively; 
\State 4. Construct the channel $\boldsymbol{H}_{\mathrm{B}, i}$ based on the channel models in Section \ref{sec_model} and $\mathcal{B}^{\text{opt}}$;
\State 5. Compute the precoding matrices and the achievable rate for the given RIS phase profile $\mathcal{B}^{\text{opt}}$ by solving the optimization problem in \eqref{P:sumRateMax_pre}.
\Until{End (Monte Carlo)}
\end{algorithmic}
\end{algorithm}

\subsubsection{RIS optimization and rate evaluation} In Algorithm \ref{Alg1}, we summarize the procedure to optimize the RISs and to calculate the rate. Precisely, we initialize Algorithm \ref{Alg1} with the initial positions ${\boldsymbol{p}}_{\mathrm{U}_{i,0}}$ and the corresponding uncertainty covariance matrices $\boldsymbol{\Sigma}_{i,0}$ at time 0. Moreover, we set the fading parameters of the NLOS links. Then, several Monte Carlo iterations are run.
%Algorithm \ref{Alg1} is executed based on different location estimation study cases. 
In each run, we compute the estimated locations $\hat{\boldsymbol{p}}_{\mathrm{U}_{i,0}}$ of the UEs at time $0$ and the actual locations ${\boldsymbol{p}}_{\mathrm{U}_{i,\tau}}$ according to the random walk model described in Section \ref{sec_Framework}. Algorithm \ref{Alg1} is split in two parts. In the first, the RIS phase profiles are optimized as described in Section \ref{sec_RISPhase}. In the second, the channels are generated based on the actual positions of the users and the corresponding NLOS channel parameters. Finally, the precoding matrices and the achievable rate for the given RIS phase profiles $\mathcal{B}^{\text{opt}}$ are computed solving the optimization problem in \eqref{P:sumRateMax_pre}.

To demonstrate the advantages of localization in Phase I, we first analyze the stationary case in which the UEs do not move, i.e., ${\boldsymbol{p}}_{\mathrm{U}_{i,\tau}} = {\boldsymbol{p}}_{\mathrm{U}_{i,0}}$ and we compare the achievable rate of the following RIS optimization approaches:
\begin{itemize}
    \item \textit{Scheme 1}: random phase profile. In this case, the phases of the RISs are randomly chosen in $[0, 2\pi )$. Therefore, the RISs optimization in Algorithm \ref{Alg1} is not performed.
    \item \textit{Scheme 2}: the \emph{prior}-based approach. In this case, ${\boldsymbol{\Sigma}}_{i,0}$ is given by \eqref{PosPri}. 
    \item \textit{Scheme 3}: the proposed \emph{Phase I}-based approach. In this case, ${\boldsymbol{\Sigma}}_{i,0}$ is given by \eqref{PosPhase1}. 
    \item \textit{Scheme 4}: the two-timescale approach with ideal location estimation, i.e.,  ${\boldsymbol{\Sigma}}_{i,0} = {\boldsymbol{0}}$.
\end{itemize}

\begin{table*}%[!htb]
\centering
\caption{\textsc{Rate Performance with Different RIS Phase Profiles}}
\centering
    \begin{tabular}{c|c|c|c}
    \hline \hline
    \textbf{RIS Phase Profile} & \textbf{Sum Rate} (bits/s/Hz) & \textbf{Outage Rate: UE 1} (bits/s/Hz) & \textbf{Outage Rate: UE 2} (bits/s/Hz)\\
    \hline
    \textit{Scheme 1} & 0.46 & 0.78 & 0.08 \\
    \hline
    \textit{Scheme 2} & 2.46 & 0.27 & 0.18 \\
    \hline
    \textit{Scheme 3}: $T_{\mathrm{P}_1}B = 20$ & 4.79 & 4.48 & 3.81 \\
    \hline
    \textit{Scheme 3}: $T_{\mathrm{P}_1}B = 40$ & 4.96 & 4.56 & 4.47 \\
    \hline
    \textit{Scheme 4} & 5.43 & 5.5 & 5.36 \\
    \hline \hline
    \end{tabular}
    \label{TabSum}
\end{table*} 

In Table \ref{TabSum}, we show the average achievable rate for the four aforementioned schemes. Comparing \textit{Scheme 2} with \textit{Scheme 1}, it is clear that optimizing the RISs on the basis of even coarse UE localization information significantly outperforms the random phase configuration in terms of sum rate, which is consistent with the results in \cite{abrardo2020intelligent}. In addition, the proposed \emph{Phase I}-based RIS optimization scheme (\emph{Scheme 3}) leads to a significant performance improvement over \textit{Scheme 1} and \textit{Scheme 2}. Specifically, with the increase of $T_{\mathrm{P}_1}$ in Phase I, the average achievable rate increases as a consequence of the reduction of the UE location uncertainty. 
It is also shown in Table \ref{TabSum} that the proposed RIS optimization scheme approaches the optimal RIS phase profile configuration obtained by \emph{Scheme 4}. The marginal achievable rate improvement when $T_{\mathrm{P}_1}B = 40$ also indicates that we can use a small number of pilots for localization in Phase I, which saves resources for Phases II and III.

\begin{figure}
    \centering
    \includegraphics[width=\linewidth]{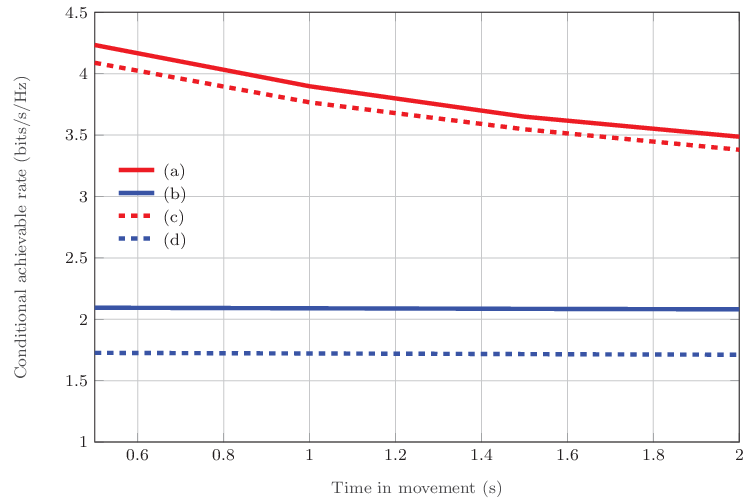}
    \caption{Achievable rate performance for mobile UEs. We assume that the UEs move in the $x$-direction. Curves (a), (b), (c), and (d) refer to \emph{Scheme 3}, \emph{Scheme 2}, \emph{Punctual - Phase I} and \emph{Punctual - Prior} schemes, respectively.} 
    \label{fig:Rate_Movement} 
    \vspace{-0.5cm}
\end{figure}

\subsubsection{Outage rate performance} Fig. \ref{fig:cdf_AR_UE2} shows the cumulative distribution function (CDF) of the achievable rate for different RIS optimization schemes. Thanks to the proposed localization-aided framework of Phase I, the location uncertainty is significantly reduced, and the CDF becomes steeper, i.e., larger rates can be supported with high probability (\emph{Scheme 3} in the figure). The achievable rates of the UE by assuming an outage probability equal to 0.1 are given in Table \ref{TabSum}. We see that the proposed RIS optimization scheme can significantly improve the outage rate.

\subsubsection{Achievable rate performance with UE movement} \label{subsubsec_mobile}
In the case of UE movement, the rate, as a function of $\tau$, of the following RIS optimization schemes is illustrated in Fig. \ref{fig:Rate_Movement}.
\begin{itemize}
  \item \emph{Scheme 2}: this is the same as Scheme 2 in Fig. \ref{fig:cdf_AR_UE2}, and it is indicated by (b) in Fig. \ref{fig:Rate_Movement}.
  \item \emph{Scheme 3}: this is the same as Scheme 3 in Fig. \ref{fig:cdf_AR_UE2}, and it is indicated by (a) in Fig. \ref{fig:Rate_Movement}. In this case, the number of pilots used in Phase I is 40. 
  \item \emph{Punctual - Phase I} (indicated by (c) in Fig. \ref{fig:Rate_Movement}). In this case, we consider Phase I with $T_{\mathrm{P}_1}B = 20$ for UE location estimation and the RIS is optimized without considering the location uncertainty, i.e., based on the \emph{punctual optimization} approach in Section \ref{sec_RISPhase}. 
  \item \emph{Punctual - Prior} (indicated by (d) in Fig. \ref{fig:Rate_Movement}). In this case, we consider the Prior scheme for UE location estimation and the RIS is optimized without considering the location uncertainty, i.e., based on the \emph{punctual optimization} approach described in Section \ref{sec_RISPhase}. 
\end{itemize}

Figure \ref{fig:Rate_Movement} shows that, owing to the increase of the location uncertainty due to the UE movements, the rate of the first two case studies degrades, indicating that Phase I needs to be implemented periodically to compensate for the movements of the UEs. The results in Fig. \ref{fig:Rate_Movement} justifies the choice of a location coherence interval equal to $T_{\mathrm{L}} = 1$s in the considered case. However, the rate degradation in case study (a) is not very significant during a single location coherence interval $T_{\mathrm{L}}$, which highlights that the proposed location-aided optimization scheme is robust to slow UE movements. From Fig. \ref{fig:Rate_Movement}, we also observe the significant performance gain of the proposed \emph{Phase I}-based fixed RIS optimization scheme over the \emph{prior}-based fixed RIS optimization scheme. 
Thus, the localization in Phase I is effective in the presence of mobile UEs, provided that the location coherence interval is optimized as a function of the mobility level of the UEs.

As for the curves (c) and (d) in Fig. \ref{fig:Rate_Movement}, we observe a clear degradation of the achievable rate. This is attributed to the misalignment between the obtained RIS profile and the actual positions of the UEs, which are different from those utilized for optimizing the RIS. When the location uncertainty is large, the probability of misalignment is more likely to occur, causing a more pronounced degradation of the achievable rate. The obtained results demonstrate that it is highly beneficial to account for the localization uncertainty when optimizing the RIS phase profiles.

In conclusion, the proposed location-aided approach for optimizing the RIS phase profile is a convenient solution, since the locations of the UE need to be estimated only every location coherence interval $T_{\mathrm{L}}$, thus reducing the overhead for configuring the RIS and the associated computational complexity. However, it is necessary that the location coherence interval $T_{\mathrm{L}}$ is adapted to the level of mobility of the UEs.

\begin{figure}
    \centering
    \includegraphics[width=\linewidth]{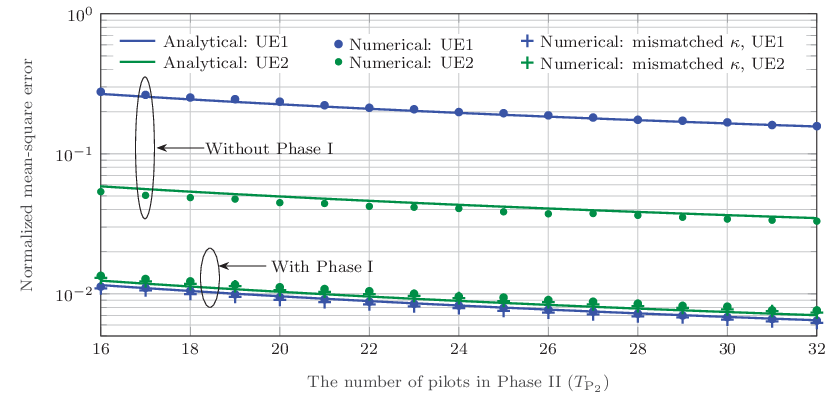}
    \caption{Channel estimation performance with a different number of pilots in Phase II. The number of pilots in Phase I is $T_{\mathrm{P}_1}B = 40$. The analytical results are obtained with \eqref{eqMMSEErr} and the Gaussian assumption is used for the derivation. The numerical results are obtained using the estimates from \eqref{eqLMMSE}. }\vspace{-5mm}
    \label{fig:CHest_GaussianAssumption}
\end{figure}

\subsection{Results for the Channel-coherent Optimization Phase}
\subsubsection{Channel estimation performance} We show the normalized mean-square error of the proposed channel estimation scheme in Fig.~\ref{fig:CHest_GaussianAssumption}. Specifically, the analytical channel estimation error is obtained as $\| \boldsymbol{E}_{i}^{(\mathrm{MMSE})} \|_{\mathrm{F}}^2/\|\boldsymbol{h}_{\mathrm{B}, i} ( \mathcal{B}, \boldsymbol{p}_{\mathrm{U}_i} ) \|^2$ where $\boldsymbol{E}_{i}^{(\mathrm{MMSE})}$ is given in \eqref{eqMMSEErr} and $\boldsymbol{h}_{\mathrm{B}, i} ( \mathcal{B}, \boldsymbol{p}_{\mathrm{U}_i} )$ is the exact channel vector. From Fig.~\ref{fig:CHest_GaussianAssumption}, we observe that the analytical formula of the channel estimation error matches well with the simulations. 

This observation justifies the Gaussian approximation for the channel estimation error, as detailed in Proposition \ref{prop1}. We also evaluate the channel estimation error when the locations of the UEs are estimated with traditional positioning techniques without executing Phase I. Both cases show a good agreement between the analytical derivation based on the Gaussian approximation and the simulations. In general, the Rician factor changes slowly when the environment is quasi static, and therefore can be estimated accurately. However, if the estimated Rician factor differs from the actual Rician factor, the performance of the proposed RIS optimization and channel estimation algorithms is negatively impacted. We analyze the sensitivity to the accurate estimation of the Rician factor in Fig.~\ref{fig:CHest_GaussianAssumption}. Specifically, we optimize the RIS by assuming a Rician factor equal to 50, while the actual Rician factor is equal to 40. In the considered case study, as can be seen from Fig.~\ref{fig:CHest_GaussianAssumption}, the mismatch of the Rician factor has a negligible impact on the accuracy of channel estimation. This result can be explained as follows. First, the performance of proposed localization algorithm does not change dramatically if the mismatch of the Rician factor is not too large. This is confirmed by the localization performance shown in Fig. \ref{fig:PEB_Kappa}. Second, since the proposed RIS optimization and channel estimation algorithms are performed by assuming a certain level of location uncertainty, the sensitivity to estimation errors of the Rician factor is reduced.

The performance gain obtained by using the proposed localization method in Phase I shows that UEs' localization, RIS optimization and channel estimation are intertwined: the localization has a beneficial effect on both RIS optimization and channel estimation. The good match of the analytical and numerical results indicates that we can use the analytical expression of the covariance of the channel estimation error to design the BS precoder, as per Section \ref{sec_PreOpt}.

\begin{figure}
    \centering
    \includegraphics[width=\linewidth]{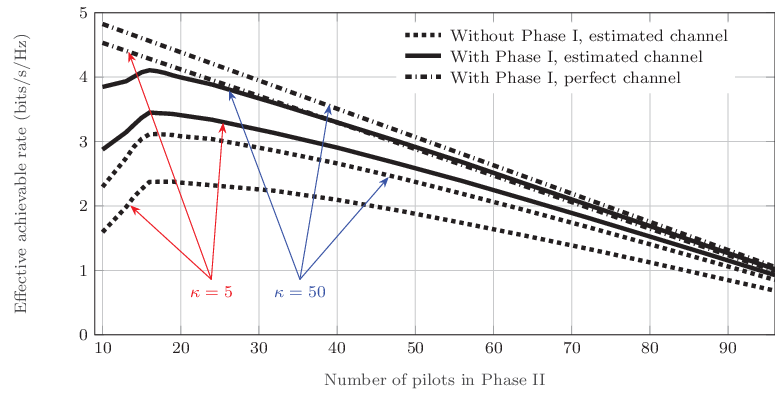}
    \caption{Effective rate of the proposed scheme. The number of pilots in Phase I is $T_{\mathrm{P}_1}B = 40$. The results are obtained using \eqref{eqARmmseCH} with different level of UE location statistics. Specifically, the proposed location-aided optimization scheme has integrated the localization in Phase I to provide better UE location information. If Phase I is not included, the UE location information is provided with traditional localization techniques. The perfect channel case is the upper bound of the proposed scheme since no channel estimation error is considered.}
    \label{fig:Rate_CHEST}
\end{figure}
\subsubsection{Effective achievable rate} We evaluate the effective achievable rate in Phase II, which is defined as $R_i^{\mathrm{eff}} (\mathcal{V}, \mathcal{B} ) = \eta R_i (\mathcal{V}, \mathcal{B} )$ where $R_i (\mathcal{V}, \mathcal{B} )$ is derived by solving the optimization problem in \eqref{P:sumRateMax_pre}, and $\eta = (T_{\mathrm{L}} - (T_{\mathrm{P}_1} + N_{\mathrm{C}} T_{\mathrm{P}_2}))/T_{\mathrm{L}} \approx (T_{\mathrm{C}} - T_{\mathrm{P}_2})/T_{\mathrm{C}}$ 
%\textcolor{blue}{Andrea: 120 is the number of OFDM symbols in each channel coherence time. Even if we have already specified it I would repeat it here} 
takes into account the overhead due to the CSI estimation. We have ignored $T_{\mathrm{P}_1}$ since $T_{\mathrm{P}_1} \ll T_{\mathrm{L}}$, in general. In Fig. \ref{fig:Rate_CHEST}, we report the effective achievable rate which is obtained by applying the optimal precoder design described in Section \ref{sec_PreOpt}. In addition, we use the conventional framework without Phase I as a benchmark. We see that the conventional channel estimation scheme has a significant performance loss compared to the proposed channel estimation scheme assisted by Phase I. The reasons are as follows. First, the localization executed in Phase I improves the RIS optimization, which creates favorable propagation channels for communication. This can be verified with the RIS optimization results in Fig.~\ref{fig:cdf_AR_UE2}. Also, the localization executed in Phase I improves the channel estimation performance, which helps to design better precoders in the presence of channel estimation error. This can be verified by the results in Fig. \ref{fig:CHest_GaussianAssumption}. The Rician factor of the RIS-UE channel has also an impact on the effective rate. Comparatively, a larger Rician factor indicates a stronger LOS path between the RISs and the UEs, hence the localization accuracy in Phase I improves. This is verified by the performance gains obtained for a larger Rician factor with respect to a smaller Rician factor, given the same number of pilots utilized in Phase II. 

The trade-off between the channel estimation accuracy and the effective achievable rate can be observed in Fig.~\ref{fig:Rate_CHEST}. To be specific, when fewer pilots are used in Phase II, the channel estimation error is large, causing the degradation of the effective achievable rate. When many pilots are used, the channel estimation error is significantly reduced, but the associated larger overhead degrades the effective achievable rate. Therefore, there is an optimal number of pilots to be used in Phase II in order to achieve the best effective achievable rate. From Fig. \ref{fig:Rate_CHEST}, when the Rician factor is 50 and the number of pilots in Phase I is 40, the optimal number of pilots in Phase II is 16 for the proposed scheme. If the Rician factor is 5, the optimal number of pilots in Phase II is 17. This indicates that more pilots are required to compensate the performance loss caused by the degradation of the localization and CSI estimation accuracy due to fading. 

Overall, from the results in Figs. \ref{R1} - \ref{fig:Rate_CHEST}, we conclude that a small number of pilots ($T_{\mathrm{P}_1}B + N_{\mathrm{C}}T_{\mathrm{P}_2}B$) is necessary during one location coherence interval to achieve the (near) optimal achievable rate in a multi-user and multi-RIS scenario, thanks to the proposed optimal RIS phase profile design, CSI estimation, and optimal precoder design. With the help of localization, we can significantly reduce the number of pilots required for RIS-assisted communications. Therefore, the proposed integrated localization and communication framework paves the way for an efficiency deployment and utilization of RISs in future wireless communications.

\section{Conclusions}\label{sec_conclusion}
In this paper, we have proposed an integrated localization and communication framework for applications to multi-user and multi-RIS wireless systems. We have proposed a new protocol consisting of three phases to perform localization, RIS optimization, CSI estimation, and precoder design. We have shown that, by integrating localization and communication, we can significantly improve the accuracy of RIS optimization based on estimates of the UE locations and the associated location uncertainty. Specifically, we have shown that the RIS phase profiles can be optimized  only sporadically based on the location coherence interval of the UEs, which is longer than the channel coherence interval for typical wireless applications. In addition, we have shown that the estimation of the CSI every channel coherent interval benefits from the proposed preceding location phase. Moreover, we have developed the optimal precoder design scheme that takes into account the instantaneous CSI and the associated estimation error. Extensive numerical results have demonstrated the effectiveness of the proposed approach.

\appendices
\section{Justification of the Gaussian Approximation} \label{GaussianAssumption}
We first obtain the channel vector $\boldsymbol{h}_{\mathrm{B}, i} ( \mathcal{B}, \boldsymbol{p}_{\mathrm{U}_i})$ as $\boldsymbol{h}_{\mathrm{B}, i} ( \mathcal{B}, \boldsymbol{p}_{\mathrm{U}_i}) =  \mathrm{vec} (\tilde{\boldsymbol{H}}_{\mathrm{B}, i}) + \sum_{k=1}^K \boldsymbol{h}_{k, i} (\boldsymbol{b}_k, \boldsymbol{p}_{\mathrm{U}_i})$. We then apply the Taylor series expansion of $\boldsymbol{h}_{\mathrm{B}, i} ( \mathcal{B}, \boldsymbol{p}_{\mathrm{U}_i})$ at a given $\boldsymbol{p}$ close to $\boldsymbol{p}_{\mathrm{U}_i}$, given by
\begin{align*}
& \boldsymbol{h}_{\mathrm{B}, i} ( \mathcal{B}, \boldsymbol{p}_{\mathrm{U}_i}) = \sum_{k=1}^K ( ( \boldsymbol{H}_{\mathrm{B}, k}^{\mathrm{T}} \mathrm{Diag} (\mathcal{B} ) \otimes \boldsymbol{I}_{N_{\mathrm{U}}} ) \mathrm{vec} (\tilde{\boldsymbol{H}}_{k, i}) \\ 
+ & \mathrm{vec} (\tilde{\boldsymbol{H}}_{\mathrm{B}, i}) + \bar{\boldsymbol{h}}_{k, i} (\mathcal{B}, \boldsymbol{p}) +  \left. \frac{\partial \bar{\boldsymbol{h}}_{k, i} (\mathcal{B}, \boldsymbol{p}_{\mathrm{U}_i})}{\partial \boldsymbol{p}_{\mathrm{U}_i}}\right|_{\boldsymbol{p}_{\mathrm{U}_i}=\boldsymbol{p}} (\boldsymbol{p}_{\mathrm{U}_i} - \boldsymbol{p}) \\
+ & O\left( (\boldsymbol{p}_{\mathrm{U}_i} - \boldsymbol{p})^2 \right).
\end{align*}
When $\sqrt{\mathrm{Tr} ( {\boldsymbol{\Sigma}}_{i,0} ) }$ is small, the chosen $\boldsymbol{p}$ is close to $\boldsymbol{p}_{\mathrm{U}_i}$; as a result, the higher orders $O\left( (\boldsymbol{p}_{\mathrm{U}_i} - \boldsymbol{p})^2 \right)$ can be reasonably ignored. The remaining first order approximation of $\boldsymbol{h}_{\mathrm{B}, i} ( \mathcal{B}, \boldsymbol{p}_{\mathrm{U}_i})$ consists of the summation of Gaussian components, which indicates that $\boldsymbol{h}_{\mathrm{B}, i} ( \mathcal{B}, \boldsymbol{p}_{\mathrm{U}_i})$ can be well approximated by a Gaussian distribution if $\sqrt{\mathrm{Tr} ( {\boldsymbol{\Sigma}}_{i,0} ) }$ is small.

\section{Proof of Proposition \ref{prop3}}\label{ProofProp3}
It can be verified that $\mathrm{E} \{ \tilde{\boldsymbol{n}}_i  \} = \boldsymbol{0}$. Since $\mathrm{E} \{ \boldsymbol{x}_i\boldsymbol{x}_j^{\mathrm{H}}\} = \boldsymbol{0}_{N_{\mathrm{U}}\times N_{\mathrm{U}}}$ for $i\ne j$, 
%and $\Delta \boldsymbol{h}_{\mathrm{B}, i}^{(\mathrm{MMSE})}$ is orthogonal to the MMSE channel estimate,
we have
\begin{align}
\tilde{J}_i = & \mathrm{E}   \{ \tilde{\boldsymbol{n}}_i \tilde{\boldsymbol{n}}_i^{\mathrm{H}}  \} = E_{\mathrm{s}} \mathrm{E}  \{ \Delta \boldsymbol{H}_{\mathrm{B}, i}^{(\mathrm{MMSE})} \boldsymbol{V}_i \boldsymbol{V}_i^{\mathrm{H}} ( \Delta \boldsymbol{H}_{\mathrm{B}, i}^{(\mathrm{MMSE})} )^{\mathrm{H}}  \} \nonumber \\
& \hspace{-3mm}+ E_{\mathrm{s}} \mathrm{E}  \{ \boldsymbol{H}_{\mathrm{B}, i} ( \sum_{j=1, j\ne i}^{N_{\mathrm{U}}} \boldsymbol{V}_j \boldsymbol{V}_j^{\mathrm{H}} ) \boldsymbol{H}_{\mathrm{B}, i}^{\mathrm{H}} \} + \sigma_n^2 \boldsymbol{I}_{N_{\mathrm{U}}}. \label{eqCoveqnoise}
\end{align}
Denoting $\Delta \boldsymbol{H}_{\mathrm{B}, i}^{(\mathrm{MMSE})} = [\Delta \boldsymbol{h}_{\mathrm{B}, i}^{(1)}, \Delta \boldsymbol{h}_{\mathrm{B}, i}^{(2)}, \cdots, \Delta \boldsymbol{h}_{\mathrm{B}, i}^{(N_{\mathrm{B}})}]$, the error covariance matrix in \eqref{eqMMSEErr} is
\begin{align}
\boldsymbol{E}_i^{(\mathrm{MMSE})} = \left[ 
\begin{array}{cccc}
\boldsymbol{E}_i^{(1, 1)} & \boldsymbol{E}_i^{(1, 2)} & \cdots & \boldsymbol{E}_i^{(1, N_{\mathrm{B}})}  \\
\boldsymbol{E}_i^{(2,1)} & \boldsymbol{E}_i^{(2,2)} & \cdots & \boldsymbol{E}_i^{(2,N_{\mathrm{B}})} \\
\vdots & \vdots & \ddots & \vdots \\
\boldsymbol{E}_i^{(N_{\mathrm{B}},1)} & \boldsymbol{E}_i^{(N_{\mathrm{B}},2)} & \cdots & \boldsymbol{E}_i^{(N_{\mathrm{B}},N_{\mathrm{B}})}
\end{array}
\right], \notag % \label{eqEmn}
\end{align}
where $\boldsymbol{E}_i^{(m,n)} = \mathrm{E} \{ \Delta \boldsymbol{h}_{\mathrm{B}, i}^{(m)} ( \Delta \boldsymbol{h}_{\mathrm{B}, i}^{(n)} )^{\mathrm{H}} \}$.
Let $\boldsymbol{\Pi}_i = \boldsymbol{V}_i \boldsymbol{V}_i^{\mathrm{H}}$ where
\begin{align*}
\boldsymbol{\Pi}_i = \left[ 
\begin{array}{cccc}
\Pi_i^{(1,1)} & \Pi_i^{(1,2)} & \cdots & \Pi_i^{(1,N_{\mathrm{B}})}  \\
\Pi_i^{(2,1)} & \Pi_i^{(2,2)} & \cdots & \Pi_i^{(2,N_{\mathrm{B}})} \\
\vdots & \vdots & \ddots & \vdots \\
\Pi_i^{(N_{\mathrm{B}},1)} & \Pi_i^{(N_{\mathrm{B}},2)} & \cdots & \Pi_i^{(N_{\mathrm{B}},N_{\mathrm{B}})}
\end{array}
\right],
\end{align*}
Thus, we obtain
\begin{align*}
& \hspace{-10mm}\mathrm{E} \left\{  \Delta \boldsymbol{H}_{\mathrm{B}, i}^{(\mathrm{MMSE})} \boldsymbol{V}_i \boldsymbol{V}_i^{\mathrm{H}} \left( \Delta \boldsymbol{H}_{\mathrm{B}, i}^{(\mathrm{MMSE})} \right)^{\mathrm{H}} \right\}\\
= & \sum_{m=1}^{N_{\mathrm{B}}} \sum_{n=1}^{N_{\mathrm{B}}} \Pi_i^{(m,n)} \mathrm{E} \left\{ \Delta \boldsymbol{h}_{\mathrm{B}, i}^{(m)} \left( \Delta \boldsymbol{h}_{\mathrm{B}, i}^{(n)} \right)^{\mathrm{H}}  \right\} \\
= & \sum_{m=1}^{N_{\mathrm{B}}} \sum_{n=1}^{N_{\mathrm{B}}} \Pi_i^{(m,n)} \boldsymbol{E}_i^{(m, n)}. \label{eqfirComp} \numberthis
\end{align*}
In addition, $\boldsymbol{R}_{\mathrm{B}, i}(\mathcal{B})$ can be expressed as
\begin{align*}
\boldsymbol{R}_{\mathrm{B}, i}(\mathcal{B})
= & \left[ 
\begin{array}{cccc}
\boldsymbol{R}_{\mathrm{B}, i}^{(1, 1)} & \boldsymbol{R}_{\mathrm{B}, i}^{(1, 2)} & \cdots & \boldsymbol{R}_{\mathrm{B}, i}^{(1, N_{\mathrm{B}})} \\
\boldsymbol{R}_{\mathrm{B}, i}^{(2, 1)} & \boldsymbol{R}_{\mathrm{B}, i}^{(2, 2)} & \cdots & \boldsymbol{R}_{\mathrm{B}, i}^{(2, N_{\mathrm{B}})} \\
\vdots & \vdots & \ddots & \vdots \\
\boldsymbol{R}_{\mathrm{B}, i}^{(N_{\mathrm{B}}, 1)} & \boldsymbol{R}_{\mathrm{B}, i}^{(N_{\mathrm{B}}, 2)} & \cdots & \boldsymbol{R}_{\mathrm{B}, i}^{(N_{\mathrm{B}},, N_{\mathrm{B}})}
\end{array}
\right],
\end{align*}
where $\boldsymbol{R}_{\mathrm{B}, i}^{(m, n)} = \mathrm{E}  \{ \boldsymbol{h}_{\mathrm{B}, i}^{(m)} (\mathcal{B}, \boldsymbol{p}_{\mathrm{U}_i}) (\boldsymbol{h}_{\mathrm{B}, i}^{(m)} (\mathcal{B}, \boldsymbol{p}_{\mathrm{U}_i}))^{\mathrm{H}}   \}$. The second component in \eqref{eqCoveqnoise} will be
\begin{align*}
& \hspace{-10mm} E_{\mathrm{s}} \mathrm{E}  \left\{ \boldsymbol{H}_{\mathrm{B}, i} \left( \sum_{j=1, j\ne i}^{N_{\mathrm{U}}} \boldsymbol{V}_j \boldsymbol{V}_j^{\mathrm{H}} \right) \boldsymbol{H}_{\mathrm{B}, i}^{\mathrm{H}} \right\} \\ 
= & E_{\mathrm{s}} \sum_{m=1}^{N_{\mathrm{B}}} \sum_{n=1}^{N_{\mathrm{B}}} \left( \sum_{j=1, j\ne i}^{N_{\mathrm{U}}} \Pi_j^{(m,n)} \right) \boldsymbol{R}_{\mathrm{B}, i}^{(m, n)}. \label{eqsecComp} \numberthis
\end{align*}
Inserting \eqref{eqfirComp} and \eqref{eqsecComp} into \eqref{eqCoveqnoise}, we have $\tilde{J}_i$ given in 
\eqref{eqtildeJi}. This completes the proof.
\bibliography{RIS_Integrated_Localization_communications}

\begin{thebibliography}{10}
\providecommand{\url}[1]{#1}
\csname url@samestyle\endcsname
\providecommand{\newblock}{\relax}
\providecommand{\bibinfo}[2]{#2}
\providecommand{\BIBentrySTDinterwordspacing}{\spaceskip=0pt\relax}
\providecommand{\BIBentryALTinterwordstretchfactor}{4}
\providecommand{\BIBentryALTinterwordspacing}{\spaceskip=\fontdimen2\font plus
\BIBentryALTinterwordstretchfactor\fontdimen3\font minus
  \fontdimen4\font\relax}
\providecommand{\BIBforeignlanguage}[2]{{%
\expandafter\ifx\csname l@#1\endcsname\relax
\typeout{** WARNING: IEEEtran.bst: No hyphenation pattern has been}%
\typeout{** loaded for the language `#1'. Using the pattern for}%
\typeout{** the default language instead.}%
\else
\language=\csname l@#1\endcsname
\fi
#2}}
\providecommand{\BIBdecl}{\relax}
\BIBdecl

\bibitem{jiang2021road}
W.~Jiang \emph{et~al.}, ``The road towards {6G}: A comprehensive survey,''
  \emph{IEEE Open J. Commun. Soc.}, vol.~2, pp. 334--366, 2021.

\bibitem{UusRugBol21}
M.~A. Uusitalo \emph{et~al.}, ``{6G} vision, value, use cases and technologies
  from european {6G} flagship project {Hexa-X},'' \emph{IEEE Access}, vol.~9,
  pp. 160\,004--160\,020, 2021.

\bibitem{YouWanHuaGao21}
X.~You \emph{et~al.}, ``Towards {6G} wireless communication networks: Vision,
  enabling technologies, and new paradigm shifts,'' \emph{Sci. China Inf.
  Sci.}, vol.~64, no.~1, pp. 1--74, Nov. 2021.

\bibitem{LiuLiuMuHou21}
Y.~Liu \emph{et~al.}, ``Reconfigurable intelligent surfaces: Principles and
  opportunities,'' \emph{IEEE Commun. Surveys Tuts.}, vol.~23, no.~3, pp.
  1546--1577, thirdquarter 2021.

\bibitem{basar2019wireless}
E.~Basar \emph{et~al.}, ``Wireless communications through reconfigurable
  intelligent surfaces,'' \emph{IEEE Access}, vol.~7, pp. 116\,753--116\,773,
  2019.

\bibitem{RenZapDebAloYueRosTre20}
M.~Di~Renzo \emph{et~al.}, ``Smart radio environments empowered by
  reconfigurable intelligent surfaces: How it works, state of research, and the
  road ahead,'' \emph{IEEE J. Sel. Areas Commun.}, vol.~38, no.~11, pp.
  2450--2525, Nov. 2020.

\bibitem{wan2021terahertz}
Z.~Wan \emph{et~al.}, ``Terahertz massive {MIMO} with holographic
  reconfigurable intelligent surfaces,'' \emph{IEEE Trans. Commun.}, vol.~69,
  no.~7, pp. 4732--4750, Jul. 2021.

\bibitem{BouAle21}
A.-A.~A. Boulogeorgos and A.~Alexiou, ``Coverage analysis of reconfigurable
  intelligent surface assisted {THz} wireless systems,'' \emph{IEEE Open J.
  Vehi. Technol.}, vol.~2, pp. 94--110, 2021.

\bibitem{liu2021integrated}
F.~Liu \emph{et~al.}, ``Integrated sensing and communications: Towards
  dual-functional wireless networks for {6G} and beyond,'' \emph{IEEE J. Sel.
  Areas Commun.}, vol.~40, no.~6, pp. 1728--1767, Jun. 2022.

\bibitem{wymeersch2021integration}
H.~Wymeersch \emph{et~al.}, ``Integration of communication and sensing in {6G}:
  a joint industrial and academic perspective,'' in \emph{Proc. IEEE 32nd Annu.
  Int. Symp. Pers., Indoor Mobile Radio Commun.}, Sep. 2021, pp. 1--7.

\bibitem{LimBelBerBou21}
C.~De~Lima \emph{et~al.}, ``Convergent communication, sensing and localization
  in {6G} systems: An overview of technologies, opportunities and challenges,''
  \emph{IEEE Access}, vol.~9, pp. 26\,902--26\,925, 2021.

\bibitem{BarWymMacBru21}
S.~Bartoletti \emph{et~al.}, ``Positioning and sensing for vehicular safety
  applications in {5G} and beyond,'' \emph{IEEE Commun. Mag.}, vol.~59, no.~11,
  pp. 15--21, Nov. 2021.

\bibitem{HeJiaKeyKokWymJun21}
J.~He \emph{et~al.}, ``Beyond {5G RIS mmWave} systems: Where communication and
  localization meet,'' \emph{IEEE Access}, vol.~10, pp. 68\,075--68\,084, 2022.

\bibitem{BjoWymMatPop22}
E.~{Bj{\"o}rnson} \emph{et~al.}, ``Reconfigurable intelligent surfaces: a
  signal processing perspective with wireless applications,'' \emph{IEEE Signal
  Process. Mag.}, vol.~39, no.~2, pp. 135--158, Mar. 2022.

\bibitem{StrAleSciDi21}
E.~C. Strinati \emph{et~al.}, ``Wireless environment as a service enabled by
  reconfigurable intelligent surfaces: The {RISE-6G} perspective,'' in
  \emph{Proc. Joint EuCNC \& 6G Summit}, Jun. 2021, pp. 562--567.

\bibitem{Dardari20}
D.~Dardari, ``Communicating with large intelligent surfaces: Fundamental limits
  and models,'' \emph{IEEE J. Sel. Areas Commun.}, vol.~38, no.~11, pp.
  2526--2537, Nov. 2020.

\bibitem{TanCheCheDai21}
W.~Tang \emph{et~al.}, ``Wireless communications with reconfigurable
  intelligent surface: Path loss modeling and experimental measurement,''
  \emph{IEEE Trans. Wireless Commun.}, vol.~20, no.~1, pp. 421--439, Jan. 2021.

\bibitem{LiuWuDiYua22}
R.~Liu \emph{et~al.}, ``A path to smart radio environments: An industrial
  viewpoint on reconfigurable intelligent surfaces,'' \emph{IEEE Wireless
  Commun. Mag.}, vol.~29, no.~1, pp. 202--208, Feb. 2022.

\bibitem{KeyKesGraWym20}
K.~{Keykhosravi} \emph{et~al.}, ``{SISO} {RIS}-enabled joint {3D} downlink
  localization and synchronization,'' in \emph{Proc. IEEE Int. Conf. Commun.},
  Jun. 2021.

\bibitem{ZhaZhaDiBiaHanSon20}
H.~{Zhang} \emph{et~al.}, ``Metalocalization: Reconfigurable intelligent
  surface aided multi-user wireless indoor localization,'' \emph{IEEE Wireless
  Commun. Mag.}, vol.~20, no.~12, pp. 7743--7757, Dec. 2021.

\bibitem{ZheYouMeiZha22}
B.~Zheng \emph{et~al.}, ``A survey on channel estimation and practical passive
  beamforming design for intelligent reflecting surface aided wireless
  communications,'' \emph{IEEE Commun. Surveys Tuts.}, vol.~24, no.~2, pp.
  1035--1071, Secondquarter 2022.

\bibitem{LiaCheLonHeLinHuaLiuSheRen21}
Y.-C. Liang \emph{et~al.}, ``Reconfigurable intelligent surfaces for smart
  wireless environments: channel estimation, system design and applications in
  {6G} networks,'' \emph{Sci. China Inf. Sci.}, vol.~64, no.~10, pp. 1--21,
  2021.

\bibitem{YouZheZha:20}
C.~{You} \emph{et~al.}, ``Intelligent reflecting surface with discrete phase
  shifts: Channel estimation and passive beamforming,'' in \emph{Proc. IEEE
  Int. Conf. Commun.}, Jun. 2020.

\bibitem{JenCar:20}
T.~L. {Jensen} \emph{et~al.}, ``An optimal channel estimation scheme for
  intelligent reflecting surfaces based on a minimum variance unbiased
  estimator,'' in \emph{Proc. IEEE Int. Conf. on Acoust., Speech and Signal
  Process. (ICASSP)}, May 2020, pp. 5000--5004.

\bibitem{WanFanDuaLi20}
P.~Wang \emph{et~al.}, ``Compressed channel estimation for intelligent
  reflecting surface-assisted millimeter wave systems,'' \emph{IEEE Signal
  Process. Lett.}, vol.~27, pp. 905--909, 2020.

\bibitem{HeWymJun21}
J.~He \emph{et~al.}, ``Channel estimation for {RIS}-aided {mmWave MIMO} systems
  via atomic norm minimization,'' \emph{IEEE Trans. Wireless Commun.}, vol.~20,
  no.~9, pp. 5786--5797, Sep. 2021.

\bibitem{MirAli21}
J.~Mirza and B.~Ali, ``Channel estimation method and phase shift design for
  reconfigurable intelligent surface assisted {MIMO} networks,'' \emph{IEEE
  Trans. Cogn. Commun. Netw.}, vol.~7, no.~2, pp. 441--451, Jun. 2021.

\bibitem{WeiHuaAleYue21}
L.~Wei \emph{et~al.}, ``Channel estimation for {RIS}-empowered multi-user
  {MISO} wireless communications,'' \emph{IEEE Trans. Commun.}, vol.~69, no.~6,
  pp. 4144--4157, Jun. 2021.

\bibitem{LiuYuaZha20}
H.~Liu \emph{et~al.}, ``Matrix-calibration-based cascaded channel estimation
  for reconfigurable intelligent surface assisted multiuser {MIMO},''
  \emph{IEEE J. Sel. Areas Commun.}, vol.~38, no.~11, pp. 2621--2636, Nov.
  2020.

\bibitem{LinJinMatYou21}
Y.~Lin \emph{et~al.}, ``Tensor-based algebraic channel estimation for hybrid
  {IRS}-assisted {MIMO-OFDM},'' \emph{IEEE Trans. Wireless Commun.}, vol.~20,
  no.~6, pp. 3770--3784, Jun. 2021.

\bibitem{KunMck21}
N.~K. Kundu and M.~R. McKay, ``Channel estimation for reconfigurable
  intelligent surface aided {MISO} communications: From {LMMSE} to deep
  learning solutions,'' \emph{IEEE Open J. Commun. Soc.}, vol.~2, pp. 471--487,
  2021.

\bibitem{LiuNgYua22}
C.~Liu \emph{et~al.}, ``Deep residual learning for channel estimation in
  intelligent reflecting surface-assisted multi-user communications,''
  \emph{IEEE Trans. Wireless Commun.}, vol.~21, no.~2, pp. 898--912, Feb. 2022.

\bibitem{LiuLeiZha20}
S.~Liu \emph{et~al.}, ``Deep learning based channel estimation for intelligent
  reflecting surface aided {MISO-OFDM} systems,'' in \emph{Proc. IEEE Veh.
  Technol. Conf.}, 2020, pp. 1--5.

\bibitem{ZapRenShaQiaDeb21}
A.~Zappone \emph{et~al.}, ``Overhead-aware design of reconfigurable intelligent
  surfaces in smart radio environments,'' \emph{IEEE Trans. Wireless Commun.},
  vol.~20, no.~1, pp. 126--141, Jan. 2021.

\bibitem{LuoLiJinChe21}
C.~Luo \emph{et~al.}, ``Reconfigurable intelligent surface-assisted multi-cell
  {MISO} communication systems exploiting statistical {CSI},'' \emph{IEEE
  Wireless Commun. Lett.}, vol.~10, no.~10, pp. 2313--2317, Oct. 2021.

\bibitem{GanZhoHuaZha21}
X.~Gan \emph{et~al.}, ``{RIS}-assisted multi-user {MISO} communications
  exploiting statistical {CSI},'' \emph{IEEE Trans. Commun.}, vol.~69, no.~10,
  pp. 6781--6792, Oct. 2021.

\bibitem{JiaYeCui20}
Y.~Jia \emph{et~al.}, ``Analysis and optimization of an intelligent reflecting
  surface-assisted system with interference,'' \emph{IEEE Trans. Wireless
  Commun.}, vol.~19, no.~12, pp. 8068--8082, Dec. 2020.

\bibitem{DaiZhuPanRenWan22}
J.~Dai \emph{et~al.}, ``Statistical {CSI}-based transmission design for
  reconfigurable intelligent surface-aided massive {MIMO} systems with hardware
  impairments,'' \emph{IEEE Wireless Commun. Lett.}, vol.~11, no.~1, pp.
  38--42, Jan. 2022.

\bibitem{HanTanJinWenMa19}
Y.~Han \emph{et~al.}, ``Large intelligent surface-assisted wireless
  communication exploiting statistical {CSI},'' \emph{IEEE Trans. Veh.
  Technol.}, vol.~68, no.~8, pp. 8238--8242, Aug. 2019.

\bibitem{HuDaiHanWan21}
C.~Hu \emph{et~al.}, ``Two-timescale channel estimation for reconfigurable
  intelligent surface aided wireless communications,'' \emph{IEEE Trans.
  Commun.}, vol.~69, no.~11, pp. 7736--7747, Nov. 2021.

\bibitem{ZhiPanRenWan21}
K.~Zhi \emph{et~al.}, ``Two-timescale design for reconfigurable intelligent
  surface-aided massive mimo systems with imperfect csi,'' \emph{IEEE Trans.
  Inf. Theory}, 2022.

\bibitem{Pan21}
C.~{Pan} \emph{et~al.}, ``An overview of signal processing techniques for
  {RIS/IRS}-aided wireless systems,'' \emph{arXiv:2112.05989}, Dec. 2021.

\bibitem{abrardo2020intelligent}
A.~Abrardo \emph{et~al.}, ``Intelligent reflecting surfaces: Sum-rate
  optimization based on statistical position information,'' \emph{IEEE Trans.
  Commun.}, vol.~69, no.~10, pp. 7121--7136, Oct. 2021.

\bibitem{WymDen20}
H.~{Wymeersch} and B.~{Denis}, ``Beyond {5G} wireless localization with
  reconfigurable intelligent surfaces,'' in \emph{Proc. IEEE Int. Conf.
  Commun.}, Jun. 2020.

\bibitem{HeWymKonSilJun20}
J.~{He} \emph{et~al.}, ``Large intelligent surface for positioning in
  millimeter wave {MIMO} systems,'' in \emph{Proc. IEEE Vehi. Technol. Conf.},
  May 2020.

\bibitem{ElzGueGuiAlo20}
A.~Elzanaty \emph{et~al.}, ``Reconfigurable intelligent surfaces for
  localization: Position and orientation error bounds,'' \emph{IEEE Trans.
  Signal Process.}, vol.~69, pp. 5386--5402, 2021.

\bibitem{ZhaZhaDiBiaHanSon21}
H.~{Zhang} \emph{et~al.}, ``Towards ubiquitous positioning by leveraging
  reconfigurable intelligent surface,'' \emph{IEEE Commun. Lett.}, vol.~25,
  no.~1, pp. 284--288, Jan. 2021.

\bibitem{NguGeoGra20}
C.~L. {Nguyen} \emph{et~al.}, ``Reconfigurable intelligent surfaces and machine
  learning for wireless fingerprinting localization,'' \emph{arXiv:2010.03251},
  2020.

\bibitem{DarDecGueGui:J21}
D.~Dardari \emph{et~al.}, ``{LOS/NLOS} near-field localization with a large
  reconfigurable intelligent surface,'' \emph{IEEE Trans. Wireless Commun.},
  vol.~21, no.~6, pp. 4282--4294, Jun. 2021.

\bibitem{Ell21}
S.~W. Ellingson, ``Path loss in reconfigurable intelligent surface-enabled
  channels,'' in \emph{Proc. IEEE 32nd Annu. Int. Symp. Pers., Indoor Mobile
  Radio Commun.}, Sep. 2021, pp. 829--835.

\bibitem{DegVitDiTre22}
V.~Degli-Esposti \emph{et~al.}, ``Reradiation and scattering from a
  reconfigurable intelligent surface: A general macroscopic model,'' \emph{IEEE
  Trans. Antennas Propag.}, early access, doi: 10.1109/TAP.2022.3149660, 2022.

\bibitem{joung2016channel}
J.~Joung \emph{et~al.}, ``Channel correlation modeling and its application to
  massive {MIMO} channel feedback reduction,'' \emph{IEEE Trans. Veh.
  Technol.}, vol.~66, no.~5, pp. 3787--3797, 2016.

\bibitem{ma2021model}
X.~Ma \emph{et~al.}, ``Model-driven deep learning based channel estimation and
  feedback for millimeter-wave massive hybrid {MIMO} systems,'' \emph{IEEE J.
  Sel. Areas Commun.}, vol.~39, no.~8, pp. 2388--2406, 2021.

\bibitem{keykhosravi2021multi}
K.~Keykhosravi and H.~Wymeersch, ``Multi-ris discrete-phase encoding for
  interpath-interference-free channel estimation,'' \emph{arXiv
  preprint:2106.07065}, Jun. 2021.

\bibitem{Kay:93}
S.~M. Kay, \emph{Fundamentals of statistical signal processing: estimation
  theory}.\hskip 1em plus 0.5em minus 0.4em\relax Upper Saddle River, NJ, USA:
  Prentice-Hall, Inc., 1993.

\bibitem{IndoorLoc16}
T.~V. Haute \emph{et~al.}, ``Performance analysis of multiple indoor
  positioning systems in a healthcare environment,'' \emph{Int. J. Health
  Geograph.}, vol.~15, no.~1, pp. 1--15, 2016.

\bibitem{Shi2011}
Q.~Shi \emph{et~al.}, ``An iteratively weighted {MMSE} approach to distributed
  sum-utility maximization for a {MIMO} interfering broadcast channel,''
  \emph{IEEE Trans. Signal Process.}, vol.~59, no.~9, pp. 4331--4340, Sep.
  2011.

\bibitem{Bertsekas1999}
D.~P. Bertsekas, \emph{{Nonlinear programming}}.\hskip 1em plus 0.5em minus
  0.4em\relax Athena Scientific, 1999.

\bibitem{BjoSan21}
E.~Bj\"{o}rnson and L.~Sanguinetti, ``Rayleigh fading modeling and channel
  hardening for reconfigurable intelligent surfaces,'' \emph{IEEE Wireless
  Commun. Lett.}, vol.~10, no.~4, pp. 830--834, 2021.

\end{thebibliography}

\end{document}